\documentclass[a4paper,12pt]{article}
\usepackage{amsfonts,amsthm,amsmath}
\usepackage{tikz}

\newcommand{\Nn}{\mathbb N}
\newcommand{\Aa}{\mathcal A}
\DeclareMathOperator{\Ff}{Fact}

\DeclareMathOperator{\card}{card}
\DeclareMathOperator{\PER}{PER}
\DeclareMathOperator{\PAL}{PAL}
\DeclareMathOperator{\CH}{CH}
\DeclareMathOperator{\Lynd}{Lynd}
\DeclareMathOperator{\St}{St}
\DeclareMathOperator{\Stand}{Stand}

\theoremstyle{plain}
\newtheorem{theorem}{Theorem}[section]
\newtheorem{proposition}[theorem]{Proposition}
\newtheorem{lemma}[theorem]{Lemma}
\newtheorem{corollary}[theorem]{Corollary}
\newtheorem{remark}[theorem]{Remark}
\newtheorem{example}[theorem]{Example}

\title{Sturmian words and the Stern sequence}

\author{Aldo de Luca $^{1}$\and Alessandro De Luca $^{2}$}
\date{}

\begin{document}
\bibliographystyle{plain}
\maketitle

\footnotetext[1]{Dipartimento di Matematica e Applicazioni 
Universit\`a degli Studi di Napoli Federico II,
via Cintia, Monte S.~Angelo
I-80126 Napoli, Italy, e-mail: \texttt{aldo.deluca@unina.it}}

\footnotetext[2]{DIETI,
Universit\`a di Napoli Federico II, via Claudio 21, 
I-80125 Napoli, Italy, e-mail: \texttt{alessandro.deluca@unina.it}}

\begin{abstract}
Central, standard, and Christoffel words are three strongly interrelated classes of binary finite words  which represent a finite counterpart of characteristic Sturmian words. A natural arithmetization of the theory is obtained by representing central
and Christoffel words by irreducible fractions labeling respectively two  binary trees, the Raney (or Calkin-Wilf) tree  and the Stern-Brocot tree. The sequence of denominators of the fractions in Raney's tree is the famous  Stern diatomic  numerical sequence. An interpretation of the terms $s(n)$
of Stern's sequence as lengths of Christoffel words when $n$ is odd,  and as  minimal periods of central words when $n$ is even, allows one to interpret  several results on Christoffel and central words in terms of Stern's sequence and, conversely, to obtain a new insight in the combinatorics of Christoffel and central words by using properties of Stern's sequence. 
One of our main results is a non-commutative version of  the  ``alternating bit sets theorem'' by Calkin and Wilf.
We also study the length distribution of Christoffel words corresponding to nodes of equal height in the tree, obtaining some interesting bounds and inequalities.

\vspace{2 mm}


\vspace{2 mm}


\end{abstract}

\section{Introduction}
Sturmian words are of great interest in combinatorics of infinite words since they are the most simple words
which are not ultimately periodic.  Since the seminal paper of 1940 by  Marston Morse and Gustav A. Hedlund \cite{MH}, there is a large literature on this  subject
(see, for instance, \cite[Chap.~2]{LO2}). Sturmian words can be defined in many different ways, of combinatorial 
or geometric nature.

In the theory a key role is played by characteristic (or standard) Sturmian words which can be
generated in several different ways and, in particular,  by a palindromization map $\psi$, introduced by the first author in \cite{deluca},
which maps injectively each finite word $v$ into a palindrome (cf.~Section \ref{sec:palmap}). The map $\psi$ can be naturally extended to infinite words. 
 In such a case if $v$ is any infinite binary word in which all letters occur infinitely often, one generates all characteristic Sturmian words. An infinite word is Sturmian if it has the same set of finite factors of a characteristic Sturmian word.
 The set of all  $\psi(v)$, with $v$ any  finite word on the binary alphabet $ {\cal A}=\{a, b\}$,  coincides with the set of  palindromic prefixes of all characteristic Sturmian words \cite{DM, deluca}.

 The words $\psi(v)$, called central, may be also defined in a purely combinatorial way as the set of all words having two coprime
 periods periods $p$ and $q$  such that the length $|\psi(v)|= p+q-2$. Central words $\psi(v)$   are strongly related \cite{deluca, BDL} to proper finite standard words which may be defined as $\psi(v)xy$ with $x,y\in {\cal A}$ and to proper Christoffel words  $a\psi(v)b$.
 
 Central, standard, and Christoffel words are considered in Section \ref{sec:three}. They represent  a finite counterpart of characteristic Sturmian words of great interest since there exist several faithful representations of the preceding words by trees, binary matrices, and continued fractions \cite{BDL}. These representations give a natural arithmetization of the theory. Some new results are proved at the end of  the section.
 
 As regards trees,
 we mainly refer in Section \ref{sec:four} to  the Raney tree. The  tree is a complete binary tree rooted at the fraction $\frac{1}{1}$ and any rational number  represented in a node as the irreducible fraction $\frac{p}{q}$ has  two children representing  the numbers $\frac{p}{p+q}$ and $\frac{p+q}{q}$. Every positive rational number appears exactly once in the tree. This tree is usually named in the literature the Calkin-Wilf tree  after Neil Calkin and Herbert Wilf, who considered it in their 2000 paper \cite{CW}.  However, the tree was introduced earlier by Jean Berstel and the first author \cite{BDL} as Raney tree, since they drew some ideas from a paper by George N. Raney \cite{Ra}. 
 
 The fraction $Ra(w)$ in the node of Raney's tree represented by the binary word $w$ is equal to the ratio $\frac{p}{q}$ of the periods of the central word $\psi(w)$, where  $p$ (resp., $q$) is the minimal period of  $\psi(w)$ if $w$ terminates with the letter $a$ (resp., $b$).
 
 Another very important tree which can be considered as dual of Raney tree is the Stern-Brocot tree (see, for instance, \cite{EL, GKP}). One
 can prove (cf.~\cite{BDL}) that the fraction $Sb(w)$ in the node $w$ of the Stern-Brocot tree is equal to the slope $\frac{|a\psi(w)b|_b}{|a\psi(w)b|_a}$ of the Christoffel word $a\psi(w)b$. The duality is due to the fact that  $$Sb(w)= Ra(w^{\sim}).$$
 
 The sequence formed by the denominators of the fractions labeling the Raney tree is the famous diatomic sequence introduced in 1858
 by  Moritz  A. Stern \cite{Stern}.  There exists a large literature on this sequence, that we shall simply refer to as Stern's sequence,  since its terms  admit   interpretations in several parts of combinatorics and  satisfy many surprising and beautiful properties (see, for instance, \cite{Leh, North, CW, CW1, CS, ADLReut,IU} and references therein). 
 
  In this paper we are mainly interested in the  properties of Stern's sequence which are related  to combinatorics of Christoffel and central words. In Section \ref{sec:five},  using some properties of Raney's tree,  we prove  that there exists a basic correspondence (cf.~Theorem  \ref{lem:SS}) between the values of Stern's sequence on {\em odd} integers and the lengths of Christoffel words as well as a correspondence between the values of the sequence on {\em even} integers and the minimal periods of central words. Thus there exists a strong relation between Sturmian words and Stern's sequence which strangely, with the only exception of  \cite{ADLReut},  has not been observed in the literature. 
  
   As a consequence of the previous correspondence several results on Stern's sequence can be proved by using the theory of Sturmian words and, conversely, properties of Stern's sequence can give a new insight in the combinatorics of Christoffel and central words.
   
   In Section \ref{sec:six} we show that one can compute the terms of Stern's sequence by continuants in two different ways. The first
   uses a result concerning the length of a Christoffel word $a\psi(v)b$ and  the minimal period of the central word $\psi(v)$ which can be expressed in terms of  continuants operating on the integral representation of the directive word $v$. The second is of a more arithmetical nature and uses known results on Stern's sequence.
   
   Section \ref{sec:seven} is devoted to a very interesting and unpublished theorem  of Calkin and Wilf on Stern's sequence \cite[Theorem 5]{CW1}. The  Calkin-Wilf theorem states that $s(n)$ represents for each $n$ the number of  ``alternating bit sets"  in $n$, i.e., the number of occurrences  of  subsequences (subwords) in the binary representation of $n$ belonging to the set $b(ab)^*$. We give two new proofs of the Calkin-Wilf theorem which are based on the combinatorics of Christoffel words. We also give a formula allowing to compute
   the length of the Christoffel word $a\psi(v)b$ in terms of the number occurrences of subwords  $u\in b(ab)^*$ in $bvb$. Moreover,  the minimal period of $\psi(v)$ equals the number of occurrences $u\in b(ab)^*$ in $bv_+b$, where $v_+$ is the longest
   prefix of $v$ immediately followed by a letter different from the last letter of $v$. A further formula shows that  $|a\psi(v)b|_a$ is equal to the number of initial occurrences of subwords  $u\in b(ab)^*$ in $bvb$.
   
   The main result of the section is  a theorem (cf.~Theorem~\ref{thm:occur}) showing the quite surprising result that for any  $w\in\Aa^{*}$, 
if we  consider the reversed  occurrences of words of the set $b(ab)^{*}$ as subwords in $bwb$, then  sorting these in decreasing lexicographic order, and marking the reversed   initial occurrences with $a$ and the reversed  non-initial ones with $b$, one yields the standard word
$\psi(w)ba$. This can be regarded as a non-commutative version of the Calkin-Wilf theorem.

In Section \ref{sec:eight}  we shall prove a formula (cf.~Theorem \ref{thm:cs2}) relating for each $w\in {\cal A}^*$ the length of the Christoffel word $a\psi(w)b$ with the  occurrences in $bwb$ of a certain kind of  factors whose number is weighted by the lengths of Christoffel words associated to suitable  directive words which are factors of $w$. The result is a consequence of an  interesting theorem on Stern's sequence due to Michael Coons and Jeffrey Shallit \cite{CS}.

In Section \ref{sec:nine} we study the distribution of the lengths of Christoffel words $a\psi(v)b$ of order $k$, i.e., the directive word $v$ has a
fixed length $k$. Using a property of Stern's sequence we show that the average value of the length is  $2(3/2)^k$. Moreover,  the maximal value
 given by $F_{k+1}$, where $(F_k)_{k\geq -1}$ is the Fibonacci numerical sequence, is reached if and only if $v$ is alternating, i.e., any letter in $v$ is immediately followed in $v$ by its complementary. 
 
  One of the main results of the section (cf.~Theorem~\ref{prop:noalt}) is that if $v\in {\cal A}^k$, with $k\geq 3$  is not alternating, then  $|a\psi(v)b| \leq F_{k+1}-F_{k-4}$, where
 the upper bound is reached if and only if  $v$ is an almost  alternating word. From this some  identities on Stern's sequence are obtained.
 Moreover, the number of missing lengths  for $k\geq 3$ has the lower bound $F_{k-4}$, so that it is exponentially increasing with $k$.
 Finally, we consider for each $k$ the maximal value $M_k$ of the number of Christoffel words of order $k$ having the same length.
 We prove that $M_k$ has a lower bound which is  exponentially increasing  with $k$.

\section{Preliminaries and notation}\label{sec:two}
Let $A$ be a finite non-empty alphabet and $A^*$ be the \emph{free monoid} generated by $A$. 
The elements of $A$ are usually called \emph{ letters} and those of $A^*$ \emph {words}. 
The identity
element of $A^*$  is called   \emph{empty word } and denoted by $\varepsilon$. We shall set $A^+=A^*\setminus \{\varepsilon\}$.
A  word $w\in A^+$ can be written uniquely as a sequence $w=w_1w_2\cdots w_n$, with $w_i\in A$, $i=1,\ldots,n$. The integer $n$ is called the \emph{length} of $w$ and is denoted by $|w|$. The length of $\varepsilon$ is conventionally $0$.

Let $w\in A^*$. A word $v$ is a \emph{factor} of $w$ if there exist words $r$ and $s$ such that $w=rvs$; $v$ is a \emph{proper} factor if $v\neq w$. 
If $r= \varepsilon$ (resp., $s=\varepsilon$), then $v$ is called a \emph{prefix} (resp.,~a \emph{suffix}) of $w$. 
If $w= rvs$, then $|r|+1$ is called an {\em occurrence} of the factor $v$ in $w$.
The number of all distinct occurrences of $v$ in $w$ is denoted by $|w|_v$. 

 A word $v=v_1v_2\cdots v_m$, $v_i\in A$, $i=1,\ldots,m$,  is a \emph{subword} of $w=w_1w_2\cdots w_n$ if there exists an $m$-tuple $(j_1, j_2, \ldots, j_m)$ such that  $$1\leq j_1<j_2< \cdots <j_m\leq n \ \mbox{and} \ v_h= w_{j_h},  \mbox {for all} \ \ h=1,2, \ldots,m.$$ Any $m$-tuple $(j_1, j_2, \ldots, j_m)$ for which the previous condition is satisfied is called an {\em occurrence} of the   subword $v$ in $w$.  We shall represent such an occurrence   also as a word
 $j_1j_2\cdots j_m$ on the alphabet $\{1,2, \ldots , n\}$. An occurrence  is said to be {\em initial} (resp., {\em final}) if
 $j_1=1$ (resp., $j_m= n$).   A factor $v$ of $w$ is trivially a subword of  $w$, whereas the converse is not in general true. The number of all distinct occurrences of the subword $v$ in $w$  is usually denoted by $\binom{w}{v}$
 and called the {\em binomial coefficient} of $w$ and $v$ (see  \cite[Chap. 6]{LO})).

 Let $w=w_1\cdots w_n$, $w_i\in A$, $1\leq i\leq n$.  The
\emph{reversal} of $w$ is the word $w^{\sim}= w_n\cdots w_1$.  One
defines also $\varepsilon^{\sim}=\varepsilon$.  A word is called
\emph{palindrome} if it is equal to its reversal.  We let
$\PAL(A)$, or simply $\PAL$, denote the set of all palindromes on the
alphabet $A$.

 Let $p$ be a positive integer.  A word $w=w_1\cdots w_n$, $w_i\in A$,
$1\leq i\leq n$, has \emph{period} $p$ if the following condition is
satisfied: for all integers $i$ and $j$ such that $1\leq i,j\leq n$,
$$\mbox{if }i\equiv j \pmod{p} , \mbox{ then } w_i = w_j.$$
We let $\pi(w)$ denote the {\em minimal period}  of $w$.  In the
sequel, we set $\pi(\varepsilon)=1$.
A word $w$ is said to be {\em constant} if  $\pi(w)=1$, i.e.,  $w=z^k$ with $k\geq 0$ and $z\in A$.

An \emph{infinite word} (from left to right) $w$ is just an infinite sequence of letters:
$$w=w_1w_2\cdots w_n\cdots \text{ where }w_i\in A,\,\text{ for all } i\geq 1\enspace.$$
 A (finite) \emph{factor} of  $w$ is either the empty word or any sequence $u=w_i\cdots w_j$ with $i\leq j$, i.e., a finite block of consecutive letters in $w$. If $i=1$, then $u$ is a \emph{prefix} of $w$;  for any $n$  we let $w_{[n]}$ denote  its prefix of length $n$, i.e., $w_{[n]}=w_1\cdots w_n$. The set of all infinite words over $A$ is denoted by $A^{\omega}$. The set of all factors of a finite or infinite word $w$ is denoted by  $\Ff w$.
 
 In the following we shall mainly concern with two-letter  alphabets. We let ${\cal A}$ denote the alphabet whose elements
 are the letters $a$ and $b$ that we shall identify respectively with  the digits $0$ and $1$; moreover, we totally order  ${\cal A}$ by setting $a<b$. We let  $(^-)$ denote  the automorphism of ${\cal A}^*$ defined
 by ${\bar a} =  b$ and ${\bar b} =  a$. For each $w\in {\cal A}^*$, the word $\bar w$ is called the {\em complementary} word, or simply the {\em complement},
 of $w$.

For each $w\in {\cal A}^*$ let $\langle w \rangle_2$, or simply  $\langle w \rangle$, denote the {\em standard interpretation} of $w$ as an integer at base $2$;
 conversely, for each integer $n\geq 0$, we let $[n]_2$ denote the {\em expansion of n at base} $2$. For instance, $\langle a\rangle=0$, $\langle b\rangle =0$, $\langle baaba \rangle =18$,
 and $[21]_2 = babab$.
 
 \subsection{The palindromization map}\label{sec:palmap}

 We  introduce in
${\cal A}^*$ the operator  $^{(+)} : {\cal A}^*\rightarrow \PAL$ which
maps any word $w\in {\cal A}^*$ to the word  $w^{(+)}$
defined as the shortest palindrome having the prefix $w$ (cf.~\cite{deluca}).  The palindrome
$w^{(+)}$ is called the \emph{right palindromic closure of} $w$.  If $Q$ is the
longest palindromic suffix of $w= uQ$, then one has
$ w^{(+)}=uQu^{\sim}\,.$
Let us now define the map
\[ \psi: {\cal A}^*\rightarrow \PAL ,\]
called  \emph{palindromization map} over ${\cal A}^*$,  as follows: $\psi(\varepsilon)=\varepsilon $ and for all
$u\in {\cal A} ^*$, $x\in {\cal A} $,
\[
	\psi(ux)=(\psi(u)x)^{(+)}\,.
\]
For instance, if $u=aba$, one has $\psi(a)=a$, $\psi(ab)= (\psi(a)b)^{(+)}= aba$, and $\psi(aba)= (abaa)^{(+)}= abaaba$.

The following proposition summarizes some simple but noteworthy  properties of  the palindromization map
 (cf., for instance, \cite{DJP, deluca}):
\begin{proposition}\label{prop:basicp} Let  $\psi$ be the palindromization map over ${\cal A}^*$.  
For $u,v\in {\cal A}^*$ the following  hold:
\begin{itemize}
\item[P1.]  If $u$ is  a prefix of  $v$, then $\psi(u)$ is a palindromic prefix (and suffix) of $\psi(v)$.
\item[P2.] If $p$ is a prefix of $\psi(v)$, then $p^{(+)}$ is a prefix of $\psi(v)$.
\item[P3.] Every palindromic prefix of $\psi(v)$ is of the form $\psi(u)$ for some prefix $u$ of $v$.
\item[P4.] The   palindromization map  is  injective.
\item[P5.]  $|\psi(u^{\sim})|= |\psi(u)|$.
\item[P6.] $\psi(\bar u) = \overline {\psi(u)}$.
\end{itemize}
\end{proposition}

 For any  $w\in \psi({\cal A}^*)$ the unique word  $u$ such that  $\psi(u)=w$ is called the \emph{directive word} of $w$.
 
For any $x\in {\cal A}$ let  $\mu_x$ denote  the injective endomorphism of ${\cal A}^*$
$$\mu_x: {\cal A}^*\rightarrow {\cal A}^*$$ 
defined by
\begin{equation}\label{eq:endo}
\mu_x(x)=x, \ \ \mu_x(y)= xy, \, \, \mbox{for}  \, \,  y\in {\cal A}\setminus \{x\} .
\end{equation}
If $v=x_1x_2\cdots x_n$, with $x_i\in {\cal A}$, $i=1,\ldots, n$, then we set:
$$ \mu_v=\mu_{x_1}\circ \cdots \circ \mu_{x_n}; $$
moreover, if $v=\varepsilon$  then  $\mu_{\varepsilon}$= id.

The following interesting theorem, proved by Jacques  Justin \cite{J} in the case of an arbitrary alphabet, relates the palindromization map to
morphisms $\mu_v$.

\begin{theorem}\label{thm:J} For all  $ v,u\in {\cal A}^*$,
$$ \psi(vu) = \mu_v(\psi(u))\psi(v).$$
In particular, if $x\in{\cal A}$, one has
$$\psi(xu)= \mu_x(\psi(u))x \ \ \mbox{and} \ \ \psi(vx)=\mu_v(x)\psi(v).$$
\end{theorem}

The palindromization map
  $\psi$ can be extended  to ${\cal A}^{\omega}$ as follows: let  $w\in {\cal A}^{\omega}$ be an infinite word
\[ w = w_1w_2\cdots w_n\cdots, \ \  \ w_i\in {\cal A}, \ i\geq 1.\]
By property P1 of  Proposition \ref{prop:basicp},  for all $n$, $\psi(w_{[n]})$ is a prefix of  $\psi(w_{[n+1]})$, so one  can define  the infinite word $\psi(w)$ as:
\[ \psi(w) = \lim_{n\rightarrow \infty} \psi(w_{[n]}).\]
The extended map $\psi: {\cal A}^{\omega}\rightarrow {\cal A}^{\omega}$ is injective. The word $w$ is called the \emph{directive word} of $\psi(w)$. 

As proved in \cite{deluca} an infinite word $s\in {\cal A}^{\omega}$ is a {\em characteristic Sturmian word}  if and only if $s= \psi(w)$ with $w\in {\cal A}^{\omega}$ such that  each letter $x\in {\cal A}$ occurs infinitely often in $w$.
An infinite word $s\in {\cal A}^{\omega}$ is called {\em  Sturmian}  if there exists a characteristic  Sturmian word $t$ such that $\Ff s = \Ff t$.\footnote{ If one extends  the action of palindromization map to infinite words over arbitrary finite alphabets, one can generate a wider class of words, called {\em standard  episturmian}, introduced in \cite{DJP}. 
 Some further extensions and generalizations of the palindromization map  are in \cite{adlADL, adlADL1, KREU}.}

 \begin{example}  {\em If $w = (ab)^{\omega}$, then the characteristic  Sturmian word $f=\psi( (ab)^{\omega})$ having the directive word $w$ is the famous \emph{Fibonacci word}
 \[ f = abaababaabaab\cdots\] }
 \end{example}

  \section{Central, standard, and Christoffel words}\label{sec:three}

In the combinatorics of Sturmian words a key role is played by three classes of finite words
 called {\em central, standard}, and {\em Christoffel words}.  They are closely interrelated and  satisfy remarkable structural properties.

A word $w$ is called central if $w$ has  two
periods $p$ and $q$ such that $\gcd(p,q)=1$ and $|w|= p+q-2$. 
The set of central words, usually denoted by $\PER$,  was introduced in~\cite{DM} where its main properties
were studied; in particular, it has been proved that $\PER$ is equal
to the set of the palindromic prefixes of all characteristic Sturmian words, i.e.,
$ \PER = \psi({\cal A}^*).$
 
There exist several  different characterizations of central words (see, for instance \cite{JB} and the references therein).
We recall here the following noteworthy  structural characterization~\cite{deluca,CdL}:

\begin{proposition}
	\label{Prop:uno}
	A word $w$ is central if and only if  it  is constant
 or  satisfies the equation
	$$w=w_1abw_2=w_2baw_1$$
	with $w_{1},w_{2}\in {\cal A}^*$.  Moreover, in this latter case, $w_1$
	and $w_2$ are uniquely determined central words, $p=|w_1|+2$ and $q=|w_2|+2$ are
	coprime periods of $w$, and $\min\{p,q\}$ is the minimal period of
	$w$.
\end{proposition}

Another important family of finite words is the class of finite standard words. In fact,  characteristic  Sturmian words can be equivalently defined in the following
way.  Let $c_1,\ldots,c_n,\ldots$ be any sequence of
integers such that $c_1\geq 0$ and $c_i>0$ for $i>1$.  We define,
inductively, the sequence of words $(s_n)_{n\geq -1}$, where
$$s_{-1}=b,\ s_0=a, \ \mbox{ and } \ 
	s_{n}=s_{n-1}^{c_{n}}s_{n-2} \ \mbox{ for } \ n\geq 1\,.$$
Since for any $n\geq 0$, $s_n$ is a proper prefix of $s_{n+1}$, the sequence $(s_n)_{n\geq -1}$ converges to a limit $s$ which is a
characteristic Sturmian word (cf.~\cite{LO2}).  Any characteristic
Sturmian word is obtained in this way. 
The Fibonacci word is obtained
when  $c_i=1$ for  all $i\geq 1$. 

We shall denote by $\Stand$
the set of all the words $s_n$, $n\geq -1$, of any sequence
$(s_n)_{n\geq -1}$.  Any word of $\Stand$ is called \emph{finite standard  word}, or
simply {\em standard word}.
The following noteworthy  relation exists \cite{DM} between standard and central
words:
$$ \Stand = {\cal A} \cup \PER \{ab, ba\}.$$
More precisely, the following holds (see, for instance~\cite[Propositions 4.9 and 4.10]{SC}):
\begin{proposition} \label{prop:standStu}Any standard word different from a single letter can be uniquely expressed
as $\mu_v(xy)$ with $\{x,y\}=\{a,b\}$ and $v\in {\cal A}^*$. Moreover, one has
$$\mu_v(xy)= \psi(v)xy.$$
\end{proposition}

\vspace{2 mm}

Let us set for any  $v\in {\cal A}^*$ and $x\in{\cal A}$, 
\begin{equation}\label{eq:muvi}
 p_x(v)= |\mu_v(x)|.
 \end{equation} 
 From Justin's formula one derives (cf.~\cite[Proposition 3.6]{adlZ1})  that $p_x(v)$ is the minimal period of $\psi(vx)$ and then a period of $\psi(v)$. Moreover, one has (cf.~\cite[Lemma 5.1]{adlZ2})
\begin{equation}\label{eq:minimalperiod}
 p_x(v)=  \pi(\psi(vx))= \pi(\psi(v)x) 
\end{equation}
 and $\gcd(p_x(v), p_y(v))=1$, so that 
\begin{equation}\label{eq:perpsi}
\pi(\psi(v))= \min \{p_x(v), p_y(v) \}.
\end{equation}
 Since $|\mu_v(xy)| = |\mu_v(x)|+  |\mu_v(y)|$, from Proposition \ref{prop:standStu}  and (\ref{eq:muvi}) one has
 \begin{equation}\label{eq:cent1}
 |\psi(v)|= p_x(v)+p_y(v)-2.
 \end{equation}
 The following lemma is readily derived from (\ref{eq:muvi}),

 \begin{lemma}\label{periods} For $w\in {\cal A}^*$ and $x,y\in {\cal A}$ one has 
 $$p_x(wx) = p_x(w), \ \ p_y(wx) = p_x(w)+ p_y(w), \ \mbox{for} \ \ y\in {\cal A}\setminus \{x\}.$$
 \end{lemma}

Let us now consider  the class $\CH$ of   words, introduced in 1875 by Elwin B. Christoffel  \cite{CFF} (see also~\cite{BLRS,BDR}),  usually called Christoffel words.  Let $p$ and $q$ be positive relatively prime integers such that $n= p+q$. The Christoffel word $w$ of {\em slope} $\frac{p}{q}$ is defined as $w= x_1\cdots x_n$ with 
$$ x_i =  \begin{cases}
a, & \mbox{if $ip\mod n > (i-1)p \mod n$}; \\
b, & \mbox{if $ip \mod n < (i-1)p \mod n$}.
\end{cases} $$
for $i=1,\ldots, n$ where $k \mod n$ denotes the remainder of the Euclidean division of $k$ by $n$.
The term slope given to the irreducible fraction $\frac{p}{q}$  is due to the fact that, as one easily derives from the definition,  $p = |w|_b$ and $q= |w|_a$. The words $a$ and $b$ are also Christoffel words with a respective slope $\frac{0}{1}$ and
$\frac{1}{0}$. The Christoffel words of  slope  $\frac{p}{q}$ with $p$ and $q$ positive integers are called {\em proper Christoffel words}.

The following  result \cite{BDL},
shows a basic relation existing between central and Christoffel words:
$$ \CH = a \PER b \ \cup \cal A. $$
Hence, any proper Christoffel word $w$ can be uniquely represented as $a\psi(v)b$ for a suitable $v\in {\cal A}^*$.
We say that $w$ is of {\em order} $k$ if $v\in {\cal A}^k$.

Let $\prec$ denote  the lexicographic order of  ${\cal A}^*$ and let  $\Lynd$ be the set of {\em Lyndon words} (see, for instance, \cite[Chap. 5]{LO}) of  ${\cal A}^*$ and $\St$ be the set of  (finite) factors of all Sturmian words. The following theorem summarizes some results on Christoffel words proved in \cite{BL,BDL,BDR,adlZ2}.
\begin{theorem}\label{thm:Lynd} Let $w=a\psi(v)b$ with $v\in{\cal A}^*$  be a proper Christoffel word. Then the following hold:
\begin{enumerate}
\item[1)] $\CH = \St \cap \Lynd$, i.e., $\CH$ equals the set of all  factors of Sturmian words which are Lyndon words.
\item[2)] There exist and are unique two Christoffel words
$w_1$ and $w_2$ such that $w=w_1w_2$. Moreover,  $w_1\prec w_2$, and $(w_1,w_2)$ is the standard factorization of $w$ in Lyndon words. 
\item[3)] If $w$ has the slope $\eta(w)= \frac{p}{q}$, then $|w_1|=p'$,
$|w_2|=q'$, where $p'$ and $q'$ are the respective multiplicative inverse of $p$ and $q$,$\mod |w|$. Moreover,
$p'= p_a(v)$, $q'=p_b(v)$ and $p= p_a(v^{\sim})$, $q=p_b(v^{\sim})$.
\end{enumerate} 
\end{theorem}

\begin{example} {\em Let $p=4$ and $q=7$. The Christoffel construction is represented by  the following diagram 
$$ 0\stackrel{a}{\longrightarrow} 4\stackrel{a}{\longrightarrow} 8\stackrel{b}{\longrightarrow} 1\stackrel{a}{\longrightarrow} 5\stackrel{a}{\longrightarrow} 9\stackrel{b}{\longrightarrow} 2\stackrel{a}{\longrightarrow} 6\stackrel{a}{\longrightarrow} 10\stackrel{b}{\longrightarrow} 3\stackrel{a}{\longrightarrow} 7\stackrel{b}{\longrightarrow} 0$$
so that the Christoffel word $w$ having slope $\frac{4}{7}$ is 
$$w=aabaabaabab =aub,$$
where $u=abaabaaba = \psi(aba^2)$ is the central word of length $9$ having  the two
coprime periods $p_a(v)=3$ and $p_b(v)=8$ with $v=aba^2$. The word $w$ can be uniquely factorized as $w=w_1w_2$, where $w_1$ and $w_2$ are the Lyndon words $w_1= aab$ and $w_2= aabaabab$. One has $w_1\prec  w_2$ with $|w_1|=3=p_a(v)$ and $|w_2|=8=p_b(v)$. Moreover, $w_2$ is the proper suffix of $w$ of maximal length which is a Lyndon word. Finally, 
$\psi(v^{\sim})= \psi(a^2ba)= aabaaabaa$, $p_a(v^{\sim})= 4= |w|_b$, $p_b(v^{\sim})=7=|w|_a$, and $|w|_bp_a(v)= 4\cdot 3=12 \equiv |w|_ap_b(v)= 7\cdot 8=56 \equiv 1 \mod 11$.} 
\end{example}

For any word $v\in {\cal A}^+$, we let  $v^F$ (resp. $v^L$) denote the first (resp.,  last) letter of $v$.

\begin{lemma}\label{lem:bar} For any $v\in {\cal A}^+$,
$ \pi(\psi(v^{\sim})) = |a\psi(v)b|_{{\bar v}^F}$.
\end{lemma}
\begin{proof}
 By (\ref{eq:perpsi}) and item 3) of Theorem \ref{thm:Lynd} one has
$$\pi(\psi(v^{\sim})) = \min\{p_a(v^{\sim}), p_b(v^{\sim})\}= \min\{|a\psi(v)b|_a, |a\psi(v)b|_b\}.$$
The result follows since for each $v\in {\cal A}^+$ one has  $|\psi(v)|_{{v^F}} > |\psi(v)|_{{\bar v}^F}$, as one easily derives  using Proposition \ref{Prop:uno},  observing that $(\psi(v))^F= v^F$ and
 by making induction on the lengths of central words.
\end{proof}

Let $v$ be a non-empty  word. We let  $v^-$ (resp., $^-v$) denote the word obtained from $v$ by deleting the last
(resp., first) letter. If $v$ is not constant, we  let $v_+$ (resp., $_+v$) denote the longest  prefix (resp., suffix) of $v$ which is immediately
followed (resp., preceded) by the complementary of the last (resp., first) letter of $v$. For instance, if $v= abbabab$, one has
$v^- = abbaba$, $v_+= abbab$, $^-v= bbabab$, and  $_+v = babab$.
\begin{proposition}\label{prop:+-} If  $v\in {\cal A}^*$ is not constant, then
$$|a\psi(v)b|= |a\psi(v^-)b|+ |a\psi(v_+)b|=  |a\psi(^-v)b|+ |a\psi(_+v)b|.$$
Moreover,
$$ |a\psi(v_+)b|= \pi(\psi(v)) \ \ \mbox{and} \ \ |a\psi(_+v)b| = |a\psi(v)b|_{{\bar v}^F}.$$
\end{proposition}
\begin{proof} Let $x$ be the last letter of $v$. By Justin's formula one has 
$$\psi(v) = \psi(v^-x)= \mu_{v^-}(x)\psi(v^-).$$
Now if  $y= \bar x$ one has 
$v^- \in (v_+)yx^*$, and by Proposition \ref{prop:standStu}, $\mu_{v^-}(x)= \mu_{(v_+) y}(x)= \mu_{v_+}(yx)= \psi(v_+)yx$.
Thus
\begin{equation}\label{eq:psi+}
\psi(v) =\psi(v_+)yx\psi(v^-) = \psi(v^-)xy\psi(v_+)
\end{equation}
and $|a\psi(v)b|= |a\psi(v^-)b|+ |a\psi(v_+)b|$. By Proposition \ref{Prop:uno}, $ |a\psi(v^-)b| =p$ and  $|a\psi(v_+)b|=q$
are two coprime periods such that $|\psi(v)|= p+q-2$. Since $v_+$ is a proper prefix of $v^-$ one has $q<p$ and
therefore $\pi(\psi(v))= |a\psi(v_+)b|$.

By item  P5  of Proposition \ref{prop:basicp}, one has $|\psi(v)|= |\psi(v^{\sim})|$, so that from the preceding result
$|a\psi(v^{\sim})b|= |a\psi((v^{\sim})^-)b|+ |a\psi((v^{\sim})_+)b|$. As it is readily verified, $(v^{\sim})^- = (^-v)^{\sim}$
and $(v^{\sim})_+ = (_+v)^{\sim}$, so that $|a\psi(v)b|=  |a\psi(^-v)b|+ |a\psi(_+v)b|$.  Since $ |a\psi(_+v)b| =  |a\psi((v^{\sim})_+)b|= \pi(\psi(v^{\sim}))$ the result follows from Lemma \ref{lem:bar}.
\end{proof}
\begin{corollary}\label{cor:Lyn+-} For any  non-constant $v \in{\cal A}^*$, the standard factorization of $a\psi(v)b$ in Lyndon words is
$$(a\psi(v_+)b, a\psi(v^-)b)\  \mbox{ if } \  v^L=a \  \mbox{and} \  (a\psi(v^-)b, a\psi(v_+)b) \ \mbox{ if } \  v^L=b. $$
As a consequence for any $v\in {\cal A}^+$
$$\pi(\psi(v)) = p_{v^L}(v).$$
\end{corollary}
\begin{proof} From (\ref{eq:psi+}) if $v$ terminates with $a$ (i.e., $x=a$)  one has $a\psi(v)b=a\psi(v_+)ba\psi(v^-)b$. On the contrary,
if $v$ terminates with $b$ (i.e., $x=b$)  one has $a\psi(v)b=a\psi(v^-)ba\psi(v_+)b$. Since $a\psi(v_+)b$ and $a\psi(v^-)b$ are
Christoffel words, the first result follows from item  2) of  Theorem \ref{thm:Lynd}. 

Let $v\in  {\cal A}^+$. If $v$ is constant, i.e., $v= x^h$ with $x\in {\cal A}$ and $h\geq 1$, then trivially  $p_x(x^h)= 1= \pi(\psi(x^h))$. If
$v$ is not constant, then by Proposition \ref{prop:+-} one has $\pi(\psi(v))= |a\psi(v_+)b|$. From the result proved above and
from item  3) of  Theorem \ref{thm:Lynd}, it follows $\pi(\psi(v)) = p_{v^L}(v).$
\end{proof}

\begin{corollary}\label{lemma:ineq} For any $v\in {\cal A}^+$ one has that $|a\psi(va)b|$ is less (resp., greater) than $|a\psi(vb)b|$ if and only if $v^L=a$ (resp., $v^L=b$).
\end{corollary}
\begin{proof} Let us first suppose that $v$ is  constant, i.e., $v= a^n$ or $v= b^n$, with $n>0$. In this case the result is trivial since $|a\psi(a^{n+1})b|= |a\psi(b^{n+1})b|=n+3$ and $|a\psi(a^nb)b|= |a\psi(b^na)b|= 2n+3$. Let us then suppose that $v$ is not constant. One has $(va)^- = (vb)^- = v$. If $v=ua$ with $u\in {\cal A}^*$, then one has
$(vb)_+ = (uab)_+ = u$ and $(va)_+ = (uaa)_+ = u_1$, with $u_1$ a proper prefix of $u$. By Proposition \ref{prop:+-}   one has:
$$ |a\psi(va)b| = |a\psi(v)b|+ |a\psi(u_1)b|<  |a\psi(v)b|+ |a\psi(u)b| = |a\psi(vb)b|.$$
In a similar way one proves that if $v=ub$, one has $ |a\psi(va)b|> |a\psi(vb)b|$.
\end{proof}

\begin{proposition}\label{prop:periods} For any word $v=v_1\cdots v_n$, with $n>0$, $v_i\in {\cal A}$, $i=1,\ldots,n$, one has
$$|\psi(v)|= \sum_{i=1}^n \pi(\psi(v_1\cdots v_i))= \sum_{i=1}^n|a\psi(v_i\cdots v_n)b|_{{\bar v}_i}.$$
\end{proposition}
\begin{proof} The result is trivial if $v$ is constant, i.e., $v=x^n$ with $x \in {\cal A}$. Indeed, $|\psi(x^n)|= |x^n|= n$ and for all $1\leq i \leq n$  one has $\pi(\psi(x^i)) =1$  and $|a\psi(x^i)b|_{\bar x} = |ax^ib|_{\bar x} =1$.  Let us then suppose that $v$ is not constant. The proof is
obtained by induction on the length of $v$. Let us prove the first equality. By Proposition \ref{prop:+-} one has 
$$ |a\psi(v)b|= |a\psi(v_+)b|+ |a\psi(v^-)b|= \pi(\psi(v)) + |a\psi(v^-)b|,$$
so that by induction
$$|\psi(v)|= \pi(\psi(v))+  \sum_{i=1}^{n-1} \pi(\psi(v_1\cdots v_i) = \sum_{i=1}^n \pi(\psi(v_1\cdots v_i)).$$
Let us now prove the second equality. By Proposition \ref{prop:+-} one has 
$$ |a\psi(v)b|= |a\psi(_+v)b|+ |a\psi(^-v)b|= |a\psi(v)b|_{{\bar v}_1}+|a\psi(^-v)b|,$$
so that by induction
$$ |\psi(v)|= |a\psi(v)b|_{{\bar v}_1} +|\psi(^-v)|=  |a\psi(v)b|_{{\bar v}_1} +  \sum_{i=2}^n|a\psi(v_i\cdots v_n)b|_{{\bar v}_i}.$$
which concludes the proof.
\end{proof}

\section{The Raney and the Stern-Brocot trees}\label{sec:four}

Let us consider the complete binary tree. Trivially, each path from the root to a particular node can be represented by a
word $w\in {\cal A}^*$. More precisely,  if  $w = b^{h_0}a^{h_1}b^{h_2}\cdots a^{h_{n-1}}b^{h_n}$  with $h_0, h_n\geq 0$ and
$h_i>0$, $1\leq i \leq n-1$, then the sequence of letters read from left to right gives the sequence of right and left moves in order to reach the node starting from the root. Since for every node there exists a unique path going from the root to the node, one
has that the nodes are faithfully represented by the words over ${\cal A}$.  Thus one can identify the nodes of the tree with
the binary words of ${\cal A}^*$. 

Let us now label each node of the tree with an irreducible fraction $\frac{p}{q}$, where $p$ and $q$ are positive and relatively prime integers, in the following way. The root has the label $\frac{1}{1}$. If a node has the label $\frac{p}{q}$, then the
left child has the label $\frac{p}{p+q}$ and the right child  has the label $\frac{p+q}{q}$. 

This tree was introduced by J. Berstel
and the first author in \cite{BDL} and called the {\em Raney tree} since it was implicitly contained in the work of Raney \cite{Ra}.
The Raney tree  was reconsidered in \cite{CW} and is usually referred in the literature  as the {\em Calkin-Wilf tree}.
\begin{figure}[ht]
\begin{center}
\footnotesize
\begin{tikzpicture} [level distance=15mm, 
level/.style={sibling distance=\textwidth/(2^#1)}
, level 4/.style={level distance=12mm}
]
\node {$\dfrac{1}{1}$}
	child {node {$\dfrac{1}{2}$}
		child {node {$\dfrac{1}{3}$}
			child {node {$\dfrac{1}{4}$}
				child {node {\vdots}}
				child {node {\vdots}}
			}
			child {node {$\dfrac{4}{3}$}
				child {node {\vdots}}
				child {node {\vdots}}
			}
		}
		child {node {$\dfrac{3}{2}$}
			child {node {$\dfrac{3}{5}$}
				child {node {\vdots}}
				child {node {\vdots}}
			}
			child {node {$\dfrac{5}{2}$}
				child {node {\vdots}}
				child {node {\vdots}}
			}
		}
	}
	child {node {$\dfrac{2}{1}$}
		child {node {$\dfrac{2}{3}$}
			child {node {$\dfrac{2}{5}$}
				child {node {\vdots}}
				child {node {\vdots}}
			}
			child {node {$\dfrac{5}{3}$}
				child {node {\vdots}}
				child {node {\vdots}}
			}
		}
		child {node {$\dfrac{3}{1}$}
			child {node {$\dfrac{3}{4}$}
				child {node {\vdots}}
				child {node {\vdots}}
			}
			child {node {$\dfrac{4}{1}$}
				child {node {\vdots}}
				child {node {\vdots}}
			}
		}
	};
\end{tikzpicture}
\end{center}
\caption{\small The Raney tree}
\end{figure}
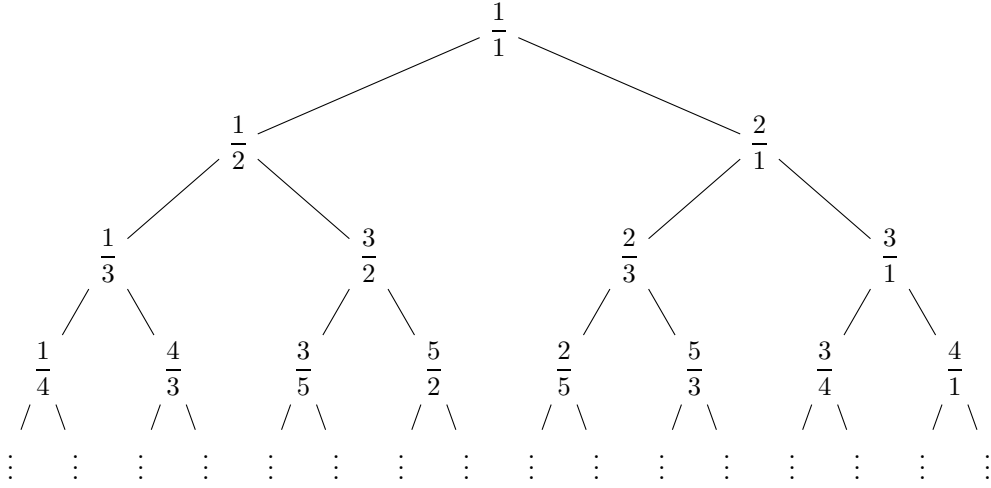

As proved in \cite{BDL} (see also \cite{CW})  all irreducible fractions can be faithfully represented by the Raney tree.  We let $Ra(w)$ denote
the fraction labeling the node represented by the word $w$.

Another famous labeling the complete binary tree by irreducible fractions  is the {\em Stern-Brocot tree} (see, for instance, \cite{EL, GKP}). The labeling is constructed as follows. The label $\frac{p}{q}$ in a node is given by $\frac{p' + p''}{q'+q''}$, where $\frac{p'}{q'}$ is the nearest ancestor above and to the left and $\frac{p''}{q''}$ is the nearest ancestor above and to the right (in order to construct the tree one needs also two {\em extra}  nodes labeled by $\frac{1}{0}$ and $\frac{0}{1}$).  We let $Sb(w)$ denote the fraction labeling the node represented by the binary word $w$. 

An important relation between the Raney and the Stern-Brocot tree is given by the following lemma\footnote{A suitable generalization of the Raney and of Stern-Brocot tree in the case of alphabets with more than two letters and of Lemma \ref{SBT}  is in \cite{adlZ2}.} (see, for instance, \cite{BDL})

\begin{lemma}\label{SBT} For  all $w\in {\cal A}^*$, one has $Sb(w) = Ra(w^{\sim})$.
\end{lemma}
\noindent Moreover, the following hold:
\begin{lemma}\label{lemma:complement} For  all $w\in {\cal A}^*$, 
$$ Ra(\bar w) = \frac{1}{Ra(w)} \ \  \mbox{and}  \ \  Sb(\bar w) = \frac{1}{Sb(w)}.$$
\end{lemma}

The  Raney  and Stern-Brocot numbers  $Ra(w)$ and $Sb(w)$ are strictly related respectively to the ratio of periods of
the central word $\psi(w)$ and to the slope of Cristoffel word $a\psi(w)b$ as follows \cite{BDL}:

\begin{proposition}\label{prop:raney} Let $w$ be the directive word of the central word $\psi(w)$. Then
$$ Ra(w) = \frac{p_a(w)}{p_b(w)}, \ \ Sb(w)= \frac{|a\psi(w)b|_b}{|a\psi(w)b|_a}. $$
\end{proposition}

By the preceding proposition and Lemma \ref{SBT},  one derives the following important duality property expressed in item 3) of Theorem \ref{thm:Lynd}

$$ |a\psi(w)b|_b= p_a(w^{\sim}), \ \ |a\psi(w)b|_a= p_b(w^{\sim}).$$
By Lemma \ref{periods} one has for any $w\in {\cal A}^*$
\begin{equation}\label{eq:ra1} 
Ra(wa)= \frac{p_a(w)}{p_a(w)+p_b(w)}, \ \  Ra(wb)= \frac{p_a(w)+p_b(w)}{p_b(w)}.
\end{equation}
Finally, we mention that both the trees can be viewed as specializations of a tree formed by ordered pairs of binary words, called
the {\em Christoffel tree} \cite{BDL, ADLReut}.
\section{The Stern sequence}\label{sec:five}

Let us enumerate the nodes of the complete binary tree as follows. The root $\varepsilon$ is numbered by $2$. If a node
$w$ has the number $\nu(w)$, then the left child $wa$ has the number $\nu(wa)= 2\nu(w)-1$ and the right child  $wb$
has the number $\nu(wb)= 2\nu(w)$. The numbering $\nu$ is a bijection of ${\cal A}^*$ into the set $\Nn_2$ of all integers
$\geq 2$ having for all $w\in {\cal A}^*$ and $n>1$
$$ \nu(w) = \langle bw \rangle +1\ \  \mbox{and}  \ \ \nu^{-1}(n) = b^{-1}[n-1]_2,$$
where $b^{-1}[n-1]_2$ is the word obtained by cancelling the first digit in the binary expansion of $n-1$. Let us set for $n>1$
$$ ra(n) = Ra(\nu^{-1}(n)),$$
and let  $s(n)$ denote  the denominator of the fraction $ra(n)$. By induction on $n$ and by (\ref{eq:ra1}) one derives:
\begin{equation}\label{eq:ra}
ra(n)= \frac{s(n-1)}{s(n)}.
\end{equation}
The sequence $s(n)$ is the famous {\em Stern sequence} \cite{Stern}  which can be inductively
defined as: $s(0)=0$, $s(1)=1$ and for  $n\geq 1$,
$$ \left \{\begin{array}{c}
s(2n) = s(n) \\
s(2n+1)= s(n)+s(n+1).
\end{array} \right .$$

The first few terms of Stern's sequence are (cf.~sequence A2487 in \cite{Sloane}):
$$ 0, 1, 1, 2, 1, 3, 2, 3, 1, 4, 3, 5, 2, 5, 3, 4, 1, 5, 4, 7, 3, 8, 5, 7, 2, 7, 5, 8, 3, 7, 4, 5, 1,\ldots$$

There exists a large literature on Stern's sequence since it satisfies a great number of beautiful and surprising properties mainly for what concerns various combinatorial interpretations that can be associated to its terms. In this paper we are mainly interested in the combinatorial properties which are related  to Christoffel and central words.

In the following  for any $w\in {\cal A}^*$,  we shall set $${\hat s} (w) = s(\nu(w)). $$

From Lemma \ref{lemma:complement} one easily derives a well known identity on Stern's sequence (see, for instance, \cite{North})

\begin{lemma} For any $k\geq 0$ and $1\leq p \leq 2^k$,
$$ s(2^k+p) = s(2^{k+1}-p).$$
\end{lemma}
\begin{proof} Let $k\geq 0$ and $w\in {\cal A}^k$. By (\ref{eq:ra}) one has $Ra(w)= \frac{s(\nu(w)-1)}{s(\nu(w))}$. By Lemma~\ref{lemma:complement},  $Ra(\bar w) = \frac{1}{Ra(w)}$ so that 
$$s(\nu(\bar w)-1) = s(\nu(w)).$$
Setting $\nu(w)= 2^k+p$, with $1\leq p \leq 2^k$, one has $\nu(\bar w)= 2^{k+1}-p+1$ and the result follows.
\end{proof}

Let us  observe that there exists a basic correspondence between the values of Stern's sequence on {\em odd} integers and the lengths of Christoffel words as well as a correspondence between the values of the sequence on {\em even} integers and the minimal periods of central words. More precisely the following hold:

\begin{theorem}\label{lem:SS} For any $w\in {\cal A}^*$ one has
$${\hat s}(wa)= s(\langle bwb\rangle) = |a\psi(w)b|, \ \ {\hat s}(wb)= s(\langle bwb\rangle+1) = \pi(\psi(wb)).$$
\end{theorem}
\begin{proof}
 By (\ref{eq:cent1}) one has $|a\psi(w)b|= 2+ |\psi(w)| = p_a(w)+ p_b(w)$. Moreover,
$${\hat s}(wa)= s(\nu(wa))= s(2\nu(w)-1) = s(\langle bwb\rangle).$$
By (\ref{eq:ra1}) one has ${\hat s}(wa)= p_a(w)+ p_b(w)$, so that the first result follows.
By  (\ref{eq:ra1}) and (\ref{eq:minimalperiod}) one has $${\hat s}(wb)=s(\nu(wb))= s(2\nu(w))=s(\langle bwb\rangle+1)= s(\nu(w))=p_b(w) =\pi(\psi(wb)),$$
which proves the second assertion of the proposition.
\end{proof}

As a consequence of the previous correspondence several results on Stern's sequence can be proved by using the theory
of Sturmian words and, conversely, properties of Stern's sequence can give a new insight in the combinatorics of Christoffel and central words.

\begin{proposition} For each $k\geq 3$ and  $0\leq p \leq 2^{k-3}-1$ one has:
$$s(2^k+8p+1) < s(2^k+8p+3) \ \mbox{and} \ \ s(2^k+8p+5) > s(2^k+8p+7).$$
\end{proposition}
\begin{proof} By Theorem \ref{lem:SS}, one has $|a\psi(va)b|= s(\langle bvab\rangle)$ and $|a\psi(vb)b|= s(\langle bvbb\rangle)$; moreover, $\langle bvbb\rangle = \langle bvab\rangle +2$. If
$v=ua$, with $u\in {\cal A}^*$, one has $\langle bvab\rangle= 1+ 8\langle u \rangle + 2^{|u|+3}$. Let us set $|u|+3=k$; one has $0\leq \langle u \rangle < 2^{k-3}$. Thus the first inequality follows from Corollary \ref{lemma:ineq}. If $v=ub$ one derives the second inequality in a similar way. 
\end{proof}

Let us define a function $R:\Nn\to\Nn$ by $R(n)=\langle [n]_{2}^{\sim}\rangle$, i.e.,  $R(n)$ is  the integer obtained by reversing the binary expansion of $n$. Note that, with the exception of $R(0)=0$, $R(n)$ is always odd.

Let $n=\sum_{i=0}^{\ell}d_{i}2^{\ell-i}>0$, with $\ell=\lfloor\log_{2}n\rfloor$ and
$d_{i}\in\{0,1\}$ for all $i$, so that $d_{i}$ is the $(i+1)$th binary digit of $n$. By 
definition, one derives
$d_{\ell-i}=\left\lfloor n/2^{i}\right\rfloor
- 2 \left\lfloor n/2^{i+1}\right\rfloor$ for $0\leq i\leq\ell$, so that
\[R(n)=\sum_{i=0}^{\ell} \left(\left\lfloor\frac{n}{2^{i}}\right\rfloor
- 2 \left\lfloor\frac{n}{2^{i+1}}\right\rfloor\right) 2^{\ell-i}
= 2^{\ell}n - 3\sum_{i=1}^{\ell}\left\lfloor\frac{n}{2^{i}}\right\rfloor 2^{\ell-i}\,.\]

The following known result (see, for instance, \cite{IU, North}) can be simply proved  using Proposition~\ref{prop:basicp}
and Theorem~\ref{lem:SS}.
\begin{proposition}
For all $n\geq 0$, the identity $s(n)=s(R(n))$ holds.
\end{proposition}

\begin{proof}
The assertion is trivial if $n=0$ or if $n$ is a power of 2, so we can write
$[n]_{2}=bwba^{k}$ for some $w\in\Aa^{*}$ and $k\geq 0$. By Theorem~\ref{lem:SS}, 
Proposition~\ref{prop:basicp}, and the definition of $s$, we have
\[
\begin{split}
s(R(n)) &= s(\langle a^{k}bw^{\sim}b\rangle)=s(\langle bw^{\sim}b\rangle)
=|a\psi(w^{\sim})b|\\
&=|a\psi(w)b|=s(\langle bwb\rangle)=s(2^{k}\langle bwb\rangle)=s(\langle bwba^{k}\rangle)
=s(n)
\end{split}\]
as desired.
\end{proof}

For each $n>1$, let us set  $L(n)=\lceil \log_2 n\rceil-1$ 
and define for  $1\leq k \leq  L(n)$,
$$\delta_k(n)= \left \lfloor \frac{n - 2^{L(n)}-1}{2^{L(n) -k}} \right \rfloor -  \left \lfloor \frac{n - 2^{L(n)}-1}{2^{L(n) -k+1}} \right \rfloor.$$

\begin{proposition} For $n>1$,
$$s(2n-1)= 2 + \sum_{k=1}^{L(n)} s(2^{k-1}+ \delta_k(n)).$$
\end{proposition}
\begin{proof} For any $n>1$ there exists a word $w$ such that $2n-1= \langle bwb \rangle$ and $|w|=m = L(n)$. If $n=2$, then since
$L(2)=0$ the result is trivially verified. Let us then suppose $n>2$.
By Theorem \ref{lem:SS} and Proposition \ref{prop:periods}  one has 
$$s(\langle bwb \rangle)= |a\psi(w)b| = 2+|\psi(w)|= 2+\sum_{k=1}^m \pi(\psi(w_1\cdots w_k)).$$
By Corollary \ref{cor:Lyn+-}, Proposition \ref{prop:raney}, and (\ref{eq:ra}) one derives
$$\pi(\psi(w_1\cdots w_k))= p_{w_k}(w_1\cdots w_k) = s(\nu(w_1\cdots w_k) + \langle w_k \rangle -1) = $$
$$s(\langle bw_1\cdots w_k\rangle + \langle w_k \rangle) = s(2^k +\langle w_1\cdots w_k\rangle +  \langle w_k \rangle).$$
Since $$\langle w_1\cdots w_k\rangle = \left \lfloor \frac{\langle w \rangle}{2^{m-k}} \right \rfloor  \ \mbox{and} \ \ \langle w_k \rangle = \left \lfloor \frac{\langle w \rangle}{2^{m-k}} \right \rfloor-2 \left \lfloor \frac{\langle w \rangle}{2^{m-k+1}} \right \rfloor,$$
one obtains, as $s(2x)=s(x)$, for $x\geq 0$,
$$\pi(\psi(w_1\cdots w_k))= s\left (2^{k-1}+ \left \lfloor \frac{\langle w \rangle}{2^{m-k}} \right \rfloor- \left \lfloor \frac{\langle w \rangle}{2^{m-k+1}} \right \rfloor \right).$$
Since  $\langle w \rangle = n-2^m-1$ and $m= L(n)$, the result follows.
\end{proof}

\section{Stern's sequence and continuants}\label{sec:six}

Any word $v\in {\cal A}^*$ can be uniquely represented as:
$$v=b^{a_0}a^{a_1}b^{a_2}\cdots a^{a_{n-1}}b^{a_n},$$
where $n$ is an even integer, $a_i>0$, $i=1,\dots,n-1$, and $a_0\geq 0$,
$a_n\geq 0$. 
 We call the list $(a_0, a_1,\dots,a_n)$ the
\emph{integral representation} of the word $v$.
 If  $a_n=0$ the list
$(a_0, a_1,\dots,a_{n-1})$ is called the {\em reduced}  integral representation
of $v$.

  We can
identify the word $v$ with its integral representation and  write
$v\equiv(a_0, a_1,\dots, a_n)$. One has
$$ |v| = \sum_{i=0}^{n} a_i.$$
For instance, the words $v_1= b^2aba^2$ and $v_2= a^3bab^2$ have the integral representations
$v_1\equiv (2,1,1,2,0)$ and $v_2\equiv (0,3,1,1,2)$. The empty word $\varepsilon$ has the integral representation $\varepsilon \equiv (0)$.

The following proposition (cf.~\cite{BDL}), called also {\em mirror formula}, permits to represent the Stern-Brocot and the Raney number
of a word $v\in {\cal A}^*$ in terms of {\em continued fractions} (see, for instance, \cite{ EL, GKP})  on the elements of the integral representation of $v$.

\begin{proposition} Let $v\in {\cal A}^*$ have the integral representation  $(a_0, a_1,\dots, a_n)$. If $n=0$, $Sb(v)= Ra(v)= [a_0;1]=[a_0+1]$. 
If $n>0$, then
$$ Sb(v) = [a_0; a_1, \ldots, a_{n-1}, a_n+1] , \ \ Ra(v) = [a_n; a_{n-1}, \ldots , a_1, a_0+1]. $$
\end{proposition}

 Let  $a_0,a_1, \ldots, a_n, \ldots$ be any sequence of 
numbers.  We consider the  {\em continuant}  $K[a_0,\ldots, a_n]$  defined by recurrence as: $K[ \ \ ]= 1$,
$K[a_0]= a_0$,
 and for $n\geq 1$,
\begin{equation}\label{eq:cf02}
 K[a_0,a_1,\ldots, a_n] = a_nK[a_0,a_1,\ldots, a_{n-1}]+ K[a_0,a_1,\ldots, a_{n-2}].
 \end{equation}
As it is readily verified, for any $n\geq 0$, $K[a_0,a_1,\ldots, a_n]$ is a multivariate polynomial
in the variables $a_0, a_1, \ldots, a_n$ which is obtained by starting with the product $a_0a_1\cdots a_n$ and then striking out adjacent pairs  $a_ka_{k+1}$ in all possible ways. For instance, $K[a_0,a_1,a_2,a_3,a_4]= a_0a_1a_2a_3a_4+ a_2a_3a_4+a_0a_3a_4+a_0a_1a_4+a_0a_1a_2+a_0+a_2+a_4.$

We recall (cf.~\cite{GKP,EL})  that for every $n\geq 0$,
\begin{equation}\label{eq:reverse} 
K[a_0,\ldots, a_n]= K[a_n, \ldots, a_0],
\end{equation}
 i.e., a continuant does not change its value by reversing the order of its elements; moreover, one has
$K[1^n] = F_{n-1}$, where  $1^n$ denotes the sequence of length $n$, $(1,1, \ldots, 1)$ and 
$(F_n)_{n\geq -1}$ is the Fibonacci numerical sequence defined by $F_{-1}= F_0= 1$, and for $n\geq 0$, $F_{n+1}= F_n+F_{n-1}$. 
From (\ref{eq:cf02}) one derives  
\begin{equation}\label{eq:piuno}
K[a_0,\ldots, a_n, 1] = K[a_0,\ldots, a_{n-1}, a_n+1].
\end{equation}

Any continued fraction can be expressed in terms of continuants as follows:

\begin{equation}\label{eq:cfcont}
[a_0;a_1,\ldots, a_n]= \frac{K[a_0,a_1,\ldots, a_{n}]}{K[a_1,\ldots, a_{n}]}.
\end{equation}
The following holds \cite{FiDE}:

\begin{theorem}\label{thm:cf000} Let $w=a\psi(v)b$ be a proper Christoffel word  and $(a_0, a_1,\ldots, a_n)$, $n\geq 0$, be
the  reduced integral representation of  $v$. If $n=0$, then $ |a\psi(v)b|= K[a_0+1, 1]= K[a_0+2]$ and $\pi(\psi(v))= K[1]=1$. If $n>0$, then 
$$ |a\psi(v)b|= K[a_0+1, a_1,\ldots, a_{n-1}, a_n+1]$$
and
$$ \pi(\psi(v)) =  K[a_0+1, a_1,\ldots, a_{n-2}, a_{n-1}].$$
\end{theorem}

The following proposition shows  that one can compute the values of Stern's sequence by continuants (cf.~\cite{IU}):

\begin{proposition} If $w\in {\cal A}^*$ has the integral representation $w$ $\equiv $ $(a_0, a_1, \ldots, a_n)$, then
$$ {\hat s}(w) = s(\nu(w)) = K[a_0+1, a_1, \ldots, a_{n-1}]. $$
\end{proposition}
\begin{proof} One has $s(\nu(w))=s(2\nu(w))={\hat s}(wb)$. By Proposition \ref{lem:SS}, one has ${\hat s}(wb)= \pi(\psi(wb))$. The word $wb$
has the integral representation $$wb\equiv (a_0, a_1, \ldots, a_n+1)$$ which is reduced. By Theorem \ref{thm:cf000} the result follows.
\end{proof}
\begin{example}{\em  Let $w= ab^2a$. One has $\nu(w)= 23$ and the integral representation of $w$ is $(0,1,2,1,0)$. One has
$s(23) = K[1,1,2,1] = 7$.}
\end{example}

For any $n>0$, let $e(n)$ the exponent of the highest power of $2$ dividing $n$. The sequence $e= (e(n))_{n>0}$ (cf.~the sequence A007814 in \cite{Sloane}) is 
$$ e= 010201030102010\cdots$$
It is noteworthy that the sequence $e$, called $\omega$-{\em Rauzy}  or $\omega$-{\em bonacci}  word in \cite{CMRJ}, can be expressed using the palindromization map $\psi$ acting on the infinite word $w= 0123456\cdots$ on the alphabet $\Nn$, as  $e= \psi(0123\cdots)$.  It is known \cite{IU}
that  for $n>0$,
                              $$ \left \lfloor \frac{s(n-1)}{s(n)} \right \rfloor = e(n).$$
By using a result  attributed to Moshe Newman (see, for instance,  sequence A2487 in \cite{Sloane})  the following  holds: for all $n>0$,
\begin{equation}\label{eq:esse}
    \frac{s(n)}{s(n+1)} = \frac{1}{2 e(n)+1 - \frac{s(n-1)}{s(n)} } \ .
    \end{equation}
  
      Let us now define for all $n>0$,  
    $$\zeta(n) =  (-1)^{n+1}(2e(n)+1) .$$
    The sequence $\zeta = (\zeta_n)_{n>0}$,  with  $\zeta_n = \zeta(n)$,  is
    $$\zeta = 1 (-3) 1 (-5) 1 (-3) 1 (-7) \cdots .$$

    \begin{proposition} For all $n>0$,
    $$\frac{s(n)}{s(n+1)} = (-1)^{n+1} [0; \zeta_n, \ldots, \zeta_1] .$$
    \end{proposition}
    \begin{proof} The proof is by induction on the integer $n$.  For $n=1$ one has  $\frac{s(1)}{s(2)} = [0; \zeta_1]= [0; 1] = \frac{1}{1}$.
    Suppose the formula true up to $n-1$ and prove it for $n$. By (\ref{eq:esse}) one has
       $$ \frac{s(n)}{s(n+1)} = \frac{(-1)^{n+1}}{\zeta_n - (-1)^{n+1}\frac{s(n-1)}{s(n)} }$$
       By induction
         $$ \frac{s(n)}{s(n+1)} = \frac{(-1)^{n+1}}{\zeta_n + [0; \zeta_{n-1}, \ldots, \zeta_1]} = (-1)^{n+1} [0; \zeta_n, \ldots, \zeta_1],$$
         which concludes the proof.
    \end{proof}
    \begin{theorem} For $n>1$ one has
    $$ s(n) = (-1)^{\lfloor \frac{n-1}{2}\rfloor} K[\zeta_1, \ldots, \zeta_{n-1}]. $$
       \end{theorem}
    \begin{proof} 
    By (\ref{eq:cfcont}) one has
    $$ [0; \zeta_n, \ldots, \zeta_1] = \frac{K[0, \zeta_n, \ldots, \zeta_1]}{K[\zeta_n,\ldots, \zeta_1]}.$$
    By (\ref{eq:cf02}) and (\ref{eq:reverse}), $K[0, \zeta_n, \ldots, \zeta_1]= K[\zeta_1, \ldots, \zeta_{n-1}]$, so that by the preceding proposition
    \begin{equation}\label{eq:cappa}
      \frac{s(n)}{s(n+1)} = (-1)^{n+1}\frac{K[\zeta_1, \ldots, \zeta_{n-1}]}{K[\zeta_1, \ldots, \zeta_{n}]}.
      \end{equation}
     Since $K[\zeta_1, \ldots, \zeta_{n-1}]$ and $K[\zeta_1, \ldots, \zeta_{n}]$ are relatively prime,  one has for all $n>1$, $s(n) = |K[\zeta_1, \ldots, \zeta_{n-1}]|$. Moreover, from (\ref{eq:cappa}),  $K[\zeta_1] >0$, $K[\zeta_1, \zeta_2]<0$, $K[\zeta_1,\zeta_2,\zeta_3]<0$,
     $K[\zeta_1,\zeta_2,\zeta_3, \zeta_4]>0$, etc. Hence, it follows  $K[\zeta_1, \ldots, \zeta_{n-1}] >0$ if and only if $\lfloor \frac{n-1}{2}\rfloor$ is even, which proves our result.
     \end{proof}
     \begin{example} {\em One has $s(4)= -K[\zeta_1, \zeta_2, \zeta_3]= -K[1, -3, 1] = -(1(-3)1+1+1)= 1,  s(5)= K[1, -3, 1, -5] = 1(-3)1(-5)+ 1(-3)+ 1(-5)+ 1(-5) +1 = 3.$}
     \end{example}

\section{The Calkin-Wilf theorem}\label{sec:seven}

Let us recall the following important theorem on Stern's sequence  due to Calkin and Wilf \cite[Theorem 5]{CW1}.  We shall give a new proof based on the combinatorics of Christoffel words. A second proof is a consequence of Theorem \ref{thm:occur} (see  Remark \ref{remark:CW}).

\begin{theorem}\label{thm:CW}  For each $n\geq 0$, the term  $s(n)$ of Stern's sequence is equal to the number of occurrences of the subwords
$u \in b(ab)^*$ in the binary expansion of the integer $n$.
\end{theorem}
\begin{proof} We shall first consider the case when the integer $n$ is {\em odd}. The result is trivial if $n=1$. Let us then suppose $n>1$. Letting  $bwb$ be the binary  expansion of $n$,
by Theorem \ref{lem:SS} one has $s(n) = |a\psi(w)b|$.  The proof is by induction on the length of the directive word $w$. If $w= z^p$ with $z\in {\cal A}$ and $p\geq 0$, then the number of occurrences of subwords  $u\in b(ab)^*$ in $bz^pb$ is $p+2= |a\psi(z^p)b|$, so in this case the result is trivially achieved.
Let us then suppose  that $w$ is not constant.  We can write  $w= x^hy (_+w)$ with $x,y\in {\cal A}$, $x\neq y$, and $h\geq 1$.  We have to consider two cases.

\vspace{2 mm}

\noindent Case 1. The letter $x$ is equal to $a$. Thus $bwb= ba^hb(_+w)b$. The number of  {\em non-initial}  occurrences of subwords  $u\in b(ab)^*$ in $bwb$ is equal to the number of  {\em all}  occurrences of  the subwords $u$ in $b(_+w)b$. By induction this number is equal to $|a\psi(_+w)b|$.

The number of  {\em all}  occurrences of  subwords $u$ in $ba^{h-1}b(_+w)b = b (^-w)b$ is by induction equal to $|a\psi(^-w)b|$. This number is equal to the number of the {\em initial}  occurrences of  the subwords $u$ in the word $bwb$. Indeed, recall that an occurrence of a
subword $u\in b(ab)^*$ in $v= bwb$  is a word $j_1j_2\cdots j_m$  on the alphabet $\{1,2, \ldots, k+2\}$, with $m= |u|$ and $k=|w|$,
such that $u_h=v_{j_h}$, $h= 1, \ldots, m$.
Any initial occurrence of $u$ in $bwb$  in which the symbol $2$ does not appear (i.e., $j_1=1, j_2>2$) is an initial occurrence 
of $u$ in $b(^-w)b$. Conversely,  any initial occurrence of $u$  in $b(^-w)b$ is an  initial occurrence of $u$ in $bwb$  in which the symbol $2$ does not appear. Moreover, there exists a one-to-one correspondence between the initial occurrences of $u$
in $bwb$ beginning with $12$ and the non-initial occurrences of $u$ in $b(^-w)b$.

Hence, the total number of  occurrences of subwords $u\in b(ab)^*$ in $bwb$ is given by  $|a\psi(_+w)b|+|a\psi(^-w)b|$. By Proposition \ref{prop:+-}, this
number is equal to $|a\psi(w)b|$. 

\vspace{2 mm}

\noindent Case 2.  The letter $x$ is equal to $b$.  Thus $bwb= b^{h+1}a(_+w)b$. The number of  {\em initial}  occurrences of subwords  $u\in b(ab)^*$ in $bwb$ is equal to the number of  {\em all}  occurrences of  the subwords $u$ in $b(_+w)b$. By induction this number is equal to $|a\psi(_+w)b|$.

Any {\em non-initial} occurrence of a subword $u$ in $bwb$ is a word $j_1j_2\cdots j_m$ over  the alphabet $\{1,2,\ldots, k+2\}$
with $j_1\geq 2$, so that any such occurrence is an occurrence of $u$ in $b(^-w)b$. Conversely, any occurrence of $u$ in
$b(^-w)b$ is a non-initial occurrence of $u$ in $bwb$. Hence, by induction, the number of all non-initial occurrences of subwords
$u$ in $bwb$ is given by $|a\psi(^-w)b|$ and the result is achieved also in this case  by Proposition \ref{prop:+-}.

\vspace{2 mm}

\noindent
Let us now consider the case of $s(n)$  when  $n$ is an {\em even} integer. The result is trivial if $n=0$. Let us then suppose $n>0$. We can write $n = 2^{e(n)}p$ where $e(n)$ is the highest
integer such that  $2^{e(n)}$ divides $n$. Thus $e(n)>0$, $p$ is odd,  and $s(n)= s(p)$. Since $p$ is odd, by the preceding result one has that $s(p)$ is equal to the number of occurrences of subwords $u\in b(ab)^*$ in the binary expansion $bwb$
of the integer $p$. Now $ [bwba^{e(n)}]_2 = 2^{e(n)}[bwb]_2=  2^{e(n)}p = n$. Since the number of occurrences of  the subwords
$u\in b(ab)^*$ in  $bwba^{e(n)}$ is equal to the number of occurrences of the subwords $u$ in $bwb$ the result follows.
\end{proof}

A  consequence of Theorem \ref{thm:CW} on  Christoffel  and central words is 

\begin{proposition}\label{prop:CW11} For each  $v\in {\cal A}^*$
$$ |a\psi(v)b| =  \sum_{u\in b(ab)^*}\binom{bvb}{u}.$$
 If  $v$ is not constant, then
$$  \pi(\psi(v)) = \sum_{u\in b(ab)^*}\binom
{bv_+b}{u}.$$
 
\end{proposition}
\begin{proof} For each $v$, one has by Proposition  \ref{lem:SS},  $s(2\nu(v)-1)= s(\langle bvb\rangle)= |a\psi(v)b|$, so the first equality follows from Theorem \ref{thm:CW}. The second equality is derived  from the second statement of Proposition \ref{prop:+-}. 
\end{proof}

\begin{proposition} For any $v\in {\cal A}^*$ the number of initial occurrences of the subwords $u \in b(ab)^*$ in $bvb$ is
given by $|a\psi(v)b|_a$.
\end{proposition}
\begin{proof} We shall prove equivalently  that the number of non-initial occurrences of the subwords $u \in b(ab)^*$ in $bvb$ is
given by $|a\psi(v)b|_b$. This number  equals   the number of all occurrences of  the subwords $u$ in  $vb$. 

The result is trivial if  $v $ is constant.  Let us then suppose that $v$ is not constant.
We can write $v = x^ky (_+v)$ with $k>0$ and $x,y\in {\cal A}$, $x\neq y$. We first suppose that $x=a$.  By
Proposition \ref{prop:CW11}
one has
$$ \sum_{u\in b(ab)^*}\binom
 {vb}{u} = \sum_{u\in b(ab)^*}\binom
 {a^kb(_+v)b}{u} = \sum_{u\in b(ab)^*}\binom
 {b(_+v)b}{u} = |a\psi(_+v)b|.$$
 Since by Proposition \ref{prop:+-} one has  $|a\psi(_+v)b|= |a\psi(v)b|_b $,  in this case the result is achieved.
 If $x=b$, then
 $$\sum_{u\in b(ab)^*}\binom
 {vb}{u} = \sum_{u\in b(ab)^*}\binom{
 b^{k}a(_+v)b}{u} = \sum_{u\in b(ab)^*}\binom
 {b b^{k-1}a(_+v)b}{u} = |a\psi(^-v)b|. $$
 By Proposition \ref{prop:+-}, one has  $|a\psi(^-v)b|=  |a\psi(v)b|- |a\psi(_+v)b|=  |a\psi(v)b|-  |a\psi(v)b|_a= |a\psi(v)b|_b$.
 From this the result follows.
 \end{proof}
 We have seen in Proposition \ref{prop:CW11}, as a consequence of the Calkin-Wilf theorem,  that the number of occurrences of words
$u\in b(ab)^*$ as subwords of $bwb$ is equal to $|a\psi(w)b|= |\psi(w)ba|$. The following theorem shows that distinguishing
between initial and non-initial occurrences and sorting them in a suitable way, one can construct the standard word
$\psi(w)ba$. This result can be regarded as a non-commutative version of the Calkin-Wilf theorem.

 \begin{theorem}\label{thm:occur}
Let $w\in\Aa^{*}$, 
and  consider the reversed  occurrences of words of the set $b(ab)^{*}$ as subwords in $bwb$. Sorting these in decreasing lexicographic order, and marking the reversed   initial occurrences with $a$ and the reversed  non-initial ones with $b$, yields the standard word
$\psi(w)ba$.
\end{theorem}
\begin{proof}
Let $w\in\Aa^{n}$ for some integer $n\geq 0$ and let ${\cal B}= \{1,2, \ldots, n+2\}$ the $(n+2)$-letter alphabet totally  ordered
by the natural integer order  $h<h+1$, $h=1, 2, \ldots, n+1$. This order can be extended to the lexicographic order  $\prec$ in ${\cal B}^*$. 
We say that $\psi(w)ba$ {\em describes} the reversed occurrences in $bwb$ of subwords in the set $b(ab)^{*}$ if it is generated  by the sequence of markers associated with the sequence of the previous occurrences sorted in  decreasing lexicographic order.  In what follows for simplicity we shall use the term  {\em occurrence}  instead of  {\em reversed occurrence  of words of the set $b(ab)^{*}$ as subwords in $bwb$}. Hence, an occurrence is initial if it ends with  $1$. 

If $w=a^{n}$, then clearly the desired sequence of occurrences in  decreasing lexicographic order
is $(n+2)(n+1)1 \succ (n+2)n 1 \succ  \cdots \succ (n+2)21 \succ (n+2) \succ 1$, which gives rise to the sequence of markers $a^{n}ba=\psi(a^n)ba$. If $w=b^{n}$, the result is trivial.

Let us now suppose, by induction, that the result holds for $w\in\Aa^{*}$ and
all shorter words, and prove it for both $wa$ and $wb$. In the case of $wa$, we can assume that $b$ occurs in $w$ and write
$w=w'ba^{k}$ for some $w'\in\Aa^{*}$ and $k\geq 0$. Since $(wa)_+= w'$ and $(wa)^-= w$, by (\ref{eq:psi+}) we have
\begin{equation}
\label{eq:w0}
\psi(wa)ba= \psi(w')ba\psi(w)ba=
\psi(w')ba\psi(w'ba^{k})ba . 
\end{equation}
Let $|bw'b|=h$, and observe that the occurrences in $bwab= bw'ba^{k}ab$
containing the position $h+k+1$  (i.e., the last $a$) are all greater (in lexicographic order) than the ones not containing it; moreover, such occurrences cannot contain the preceding $k$  letters $a$, and have to start with $(h+k+2)(h+k+1)$. Hence, we can identify any such occurrence $(h+k+2)(h+k+1)\alpha$, $\alpha \in {\cal B}^*$, with the occurrence $\alpha$ in $bw'b$, and identify any occurrence not containing $(h+k+1)$ as an occurrence in $bw'ba^{k}b$. These bijections preserve
the lexicographic order and the property of being initial or not.  By
induction, 
  the occurrences in $bw'b$ are described by $\psi(w')ba$ and the occurrences in $bwb$ by $\psi(w)ba$.  It follows by \eqref{eq:w0} that $\psi(wa)ba$ describes  the 
occurrences in $bwab$ of subwords in the set $b(ab)^{*}$.

Let us now consider $wb$; in this case we can assume that $a$ occurs in $w$, and write
$w=w'ab^{k}$ for some $k\geq 0$ and $w'\in\Aa^{*}$. Therefore, since  $(wb)_+= w'$ and $(wb)^- = w$, by (\ref{eq:psi+}) we have
\begin{equation}
\label{eq:w1}
\psi(wb)ba= \psi(w)ba\psi(w')ba=\psi(w'ab^{k})ba\psi(w')ba.
\end{equation}

Now, occurrences of words in $b(ab)^*$ as subwords of $bwbb=bw'abb^{k+1}$ can be divided in the following three classes, in decreasing lexicographic order:
\begin{enumerate}
\item occurrences containing a position greater than $h+1$, where $h=|bw'a|$,
\item occurrences containing positions $h+1$ and $h$ (the $ab$ right after $w'$),
\item occurrences containing position $h+1$ but not $h$.
\end{enumerate}
Members of the first class can never contain position $h+1$, so that they can naturally be identified with occurrences in $bw'ab^{k+1}$ (and so that the three classes are
disjoint). Members of the second class can also be identified with occurrences in
$bw'ab^{k+1}$ after discarding positions $h+1$ and $h$. Under such correspondences (that do not alter the lexicographic order nor the property of being initial), it is easy to see that together, the first two classes make up all occurrences of words in
$b(ab)^{*}$ as subwords of $bw'ab^{k+1}= bwb$, so that by induction hypothesis they are
described by $\psi(w)ba$.

Occurrences in the third class can obviously be seen as occurrences in $bw'b$, and are therefore described by $\psi(w')ba$; the assertion then follows by~\eqref{eq:w1}.
\end{proof}

\begin{example} {\em Let $w= abbaa$, so that $\psi(w)= ababaababaababa$. The occurrences of words  in
$b(ab)^{*}$ as subwords of $bwb= babbaab$ are: 

\noindent {\em Initial occurrences}: 1, 123, 12357, 12367, 124, 12457, 12467, 127, 157, 167. 

\noindent{\em Non-initial occurrences}:  3, 357, 367, 4, 457, 467, 7.

Sorting the reversed  occurrences in decreasing lexicographic order,  one has the standard word $\psi(w)ba$
as shown by the following diagram:

\begin{gather*}
\begin{array}{cccccccccccccccccccc}
 \text{\small $a$} && \text{\small $b$} && \text{\small $a$} && \text{\small $b$} && \text{\small $a$} && \text{\small $a$} && \text{\small $b$} && \text{\small $a$}\\
76421 &\!>\! & 764 &\!>\! & 76321 &\!>\! & 763 &\!>\! & 761 &\!>\! & 75421 &\!>\! & 754 &\!>\! & 75321
\end{array}
\\
\begin{array}{cccccccccccccccccccc}
& \text{\small $b$} && \text{\small $a$} && \text{\small $a$} && \text{\small $b$} && \text{\small $a$} && \text{\small $b$} && \text{\small $a$} && \text{\small $b$} && \text{\small $a$}\\
>\! & 753 &\!>\! & 751 &\!>\! & 721 &\!>\! & 7 &\!>\! & 421 &\!>\! & 4 &\!>\! &321 &\!>\! &3 &\!>\! & 1
\end{array}
\end{gather*}

}
\end{example}
\smallskip

\begin{remark}\label{remark:CW} {\em Theorem \ref{thm:occur} gives also a different proof of the Calkin-Wilf theorem. Indeed, from Theorem \ref{thm:occur} one derives that for any $w\in {\cal A}^*$, the total number of occurrences  of words in
$b(ab)^{*}$ as subwords of $bwb$ is equal to $|\psi(w)ba|= |a\psi(w)b|$.}
\end{remark}

\section{The Coons-Shallit theorem}\label{sec:eight}

In this section we shall prove a formula relating for each $w\in {\cal A}^*$ the length of the Christoffel word $a\psi(w)b$ with the  occurrences in $bwb$ of a certain kind of  factors whose number is weighted by the lengths of Christoffel words associated to suitable  directive words which are factors of $w$. The result is a consequence of the following interesting theorem on Stern's sequence due to Coons and Shallit \cite{CS}.

For any $n\geq 0$ and $w\in {\cal A}^*$ let $\alpha_w(n)$ simply denote the number
of occurrences of $w$ in the binary expansion of the integer $n$, i.e., $\alpha_w(n)= |[n]_2|_w$.

\begin{theorem} For any $n\geq 0$ and $w\in {\cal A}^*$, 
$$s(n) = \alpha_b(n) + \sum_{w\in b{\cal A}^*} s(\langle {\bar w}\rangle) \alpha_{wb}(n).$$
\end{theorem}

Let us now define the two following sets of words  $\Gamma_1= \{ u\in b{\cal A}^*b \mid |u|_a=1\}$ and
 $\Gamma_2= \{ u\in b{\cal A}^*b \mid |u|_a \geq 2\}$. Moreover, to each word $u\in \Gamma_2$ we can associate
 the unique word ${\hat u}$ such that $u \in b^+a{\hat u}ab^+$, i.e., ${\hat u}$ is the unique factor of $u$ between the first and the last occurrence of $a$ in $u$. The following holds:
 
 \begin{theorem}\label{thm:cs2} For any $w\in {\cal A}^*$ one has
$$|a\psi(w)b|= |bwb|_b + \sum_{u\in \Gamma_1} |bwb|_u + \sum_{u\in \Gamma_2} |a\psi({\hat u})b| |bwb|_u .$$
 \end{theorem} 
 \begin{proof}
By Proposition~\ref{lem:SS}, the Coons-Shallit  theorem implies that for any $w\in\Aa^{*}$ we have
\begin{equation}
\label{eq:cs}
|a\psi(w)b|= s(\langle bwb\rangle)= |bwb|_{b}+\sum_{u\in b\Aa^{*}}s(\langle\bar u\rangle)|bwb|_{ub}\,.
\end{equation}
As $s(\langle v\rangle)=s(2\langle v\rangle)=s(\langle va\rangle)$ for any
$v\in\Aa^{*}$, the last sum in~\eqref{eq:cs} can be replaced by
$$\sum_{u\in b\Aa^{*}b}s(\langle\bar u\rangle)|bwb|_{u}. $$
Now, the set $b\Aa^{*}b$ is
clearly a disjoint union of $b^{*}$, $\Gamma_{1}$, and $\Gamma_{2}$. Since $s(0)=0$, we
only need to calculate that sum on $\Gamma_{1}$ and $\Gamma_{2}$.

For any $u\in\Gamma_{1}$, $\langle\bar u\rangle$ is a power of 2, so that
$s(\langle\bar u\rangle)=1$ and
$$\sum_{u\in \Gamma_1}s(\langle\bar u\rangle)|bwb|_{u}=   \sum_{u\in \Gamma_1} |bwb|_u.$$
If $u\in\Gamma_{2}$, then by the properties of $s$ we have
\[s(\langle\bar u\rangle)=s(\langle\overline{a\hat ua}\rangle)=
s(\langle b\bar{\hat u}b\rangle)=|a\psi(\hat u)b|\,,\]
where the last equality comes from Theorem~\ref{lem:SS} and the fact that
$|\psi(\bar v)|=|\psi(v)|$ for any $v\in\Aa^{*}$ (cf.~item P6 of Proposition \ref{prop:basicp}). Therefore, the assertion follows
from~\eqref{eq:cs}.
\end{proof}

 \begin{example}{\em Let  $w= ababa$, so $\psi(w) = abaababaabaababaaba$ and $|a\psi(w)b|$ $=$ $21$. In $bwb= bababab$
 there is only one factor namely $bab$  beginning and terminating with $b$  and having only one occurrence of the letter $a$. One has $|bwb|_{bab}=3$.
 There are two factors $u$  in $bwb$   beginning and terminating with $b$ such that  $|u|_b\geq 2$. The first is 
 $u_1= babab$ and occurs two times in $bwb$ and the second $u_2= bababab$ occurring only once in $bwb$.  Moreover, ${\hat u}_1= b$ and ${\hat u}_2= bab$. Since $|bwb|_b= 4$, $|a\psi(b)b|=3$, and $|a\psi(bab)b|=8$, one obtains by Theorem \ref{thm:cs2}, $|a\psi(w)b| = 4+3+6+8=21$.}
 \end{example}
 
\section{Length distribution of Christoffel words}\label{sec:nine}
 We recall  that a proper Christoffel word $w$ is of {\em order} $k$, $k\geq 0$,
if  $w= a\psi(v)b$ with $v\in {\cal A}^k$. In this section we are interested in the distribution
of the lengths of Christoffel words of order $k$.

By Theorem \ref{lem:SS} one has that
\begin{equation}\label{eq:abpsi}
 \{|a\psi(v)b| \mid v\in {\cal A}^k\} = \{s(2n-1) \mid 2^k+1 \leq n \leq 2^{k+1} \}.
 \end{equation}

 \begin{lemma} For each $k\geq 0$, $$\sum_{v\in {\cal A}^k} |a\psi(v)b| = 2 \cdot 3^k .$$
 \end{lemma}
 \begin{proof} From (\ref{eq:abpsi}) one has 
 $$\sum_{v\in {\cal A}^k} |a\psi(v)b| = \sum_{n= 2^k+1}^{2^{k+1}} s(2n-1)=\sum_{n= 2^k+1}^{2^{k+1}}s(n)+\sum_{n= 2^k+1}^{2^{k+1}}s(n-1).$$
 As is well known (see, for instance, \cite{Leh}), $\sum_{n= 2^k+1}^{2^{k+1}}s(n)= 3^k$. Moreover, since $s(2^k)= s(2^{k+1})= s(1)=1$, one has $$\sum_{n= 2^k+1}^{2^{k+1}}s(n-1)= \sum_{n= 2^k+1}^{2^{k+1}}s(n)= 3^k.$$ From this the result follows.
 \end{proof}
 
 Let us observe that from the preceding lemma one has that the average length of the Christoffel words of order $k$ is  $ 2(3/ 4)^k$.

We say that a word $v\in {\cal A}^k$ is   {\em alternating} if  for $x,y\in {\cal A}$ and $x\neq y$, $v = (xy)^{\frac{k}{2}}$ if $k$ is even and $v=  (xy)^{\lfloor\frac{k}{2}\rfloor} x$ if $k$ is odd, i.e., any letter in $v$ is immediately followed by its complementary.

The following lemma, as regards the upper bound,  was proved in \cite{FiDE} as an extremal property of the Fibonacci word. A different proof is obtained from (\ref{eq:abpsi}) as a property of Stern's sequence (see, for instance, \cite{North}). As regards  the lower bound the proof  is trivial.

\begin{lemma}\label{lemma:limit} For all $v\in {\cal A}^k$ one has
$$k+2 \leq  |a\psi(v)b| \leq F_{k+1},$$
where the lower bound is reached if and only if $v $ is constant  and the upper
bound is reached if and only if  $v$ is alternating.
\end{lemma}

For each $k$ let $u_k$ be the alternating word of length $k$ beginning with the letter $a$. One has that
$$\langle bu_{k-1}b \rangle = \frac{2^{k+2}+(-1)^{k+1}}{3}, \ \ \langle b{\bar u_{k-1}}b \rangle = \frac{5\cdot 2^{k}+(-1)^k}{3}.$$
Thus by Theorem \ref{lem:SS}   and the preceding lemma, one has  (see, for instance,~\cite[Theorem 2.1]{North})
$$ F_k = |a\psi(u_{k-1})b| =  |a\psi({\bar u}_{k-1})b|= s\left (\frac{2^{k+2}+(-1)^{k+1}}{3}\right )= s \left( \frac{5\cdot 2^{k}+(-1)^k}{3} \right).$$

  In the following for each word $v\in {\cal A}^*$ we let $[v]$ denote the set  $[v]= \{v, v^{\sim}, {\bar v}, {\bar v}^{\sim} \}$.  From Proposition \ref{prop:basicp} all Christoffel words
  $a\psi(z)b$ with a directive word $z\in [v]$ have the same length.

\begin{proposition}\label{prop:noconst} If $v\in {\cal A}^k$ is not constant, then
$$|a\psi(v)b| \geq 2k+1,$$
where the lower bound is reached if and only if $v\in [ab^{k-1}]$.
\end{proposition}
\begin{proof} The proof is by induction on $k$. The result is trivial if $k=2$ and $k=3$. Let  $v\in {\cal A}^k$, $k>3$,  be non-constant.
We first suppose that $v^-$ is not constant.  By   Proposition \ref{prop:+-}, 
 as 
 $|v^-|= k-1$ and $|v_+|\geq 0$, by induction  one has
$$ |a\psi(v)b| = |a\psi(v_+)b|+ |a\psi(v^-)b| \geq  2+2(k-1)+1 = 2k+1.$$ 
Suppose now that $ |a\psi(v)b| = 2k+1$. From the preceding equation one has $2k+1 \geq  |a\psi(v_+)b|+ 2(k-1)+1$, so that $|a\psi(v_+)b|\leq 2$ that implies $v_+= \varepsilon$ and $|a\psi(v^-)b|= 2(k-1)+1$. By induction $v^- \in [ab^{k-2}]$. Since $v_+= \varepsilon$,  one derives that either
$v= ab^{k-1}$ or $v = ba^{k-1}$. 

Let us now suppose that $v^-$ is constant, i.e., $v^-= x^{k-1}$ with $x\in {\cal A}$. One has $v= x^{k-1}y$ with $y= \bar x$ and
$|a\psi(v)b| = |a x^{k-1}yx^{k-1}b| = 2k+1$. It follows that in all cases  the lower bound is reached if and only if $v\in [ab^{k-1}]$.
\end{proof}

 Let us now introduce for each $k\geq 3$ the word $v_k$ as follows: 
  
      $$v_k= 
        \begin{cases}
         ab^2(ab)^{\frac{k-3}{2}},  & \mbox{ if $k$ is odd}; \\
          ab^2(ab)^{\lfloor\frac{k-3}{2}\rfloor}a, &  \mbox{ if $k$ is even}.
          \end{cases} $$ 
  
  Note that each letter of $v_k$ but the second one,  is immediately followed by its complementary. One has that $v_{k+1} = v_ka$ if $k$ is odd and $v_{k+1} = v_kb$ if $k$ is even.
  \begin{lemma}\label{upp:1} For each $k\geq 3$,
  $$ |a\psi(v_k)b| = F_{k+1} - F_{k-4}.$$
  \end{lemma}
  \begin{proof} The proof is by induction on the value of $k$. The result is true for $k=3$ and $k=4$. Indeed,
  $|a\psi(ab^2)b| = 7 = F_4-F_{-1}$ and $|a\psi(ab^2a)b|= 12= F_5-F_0$. Let us take $k>4$.  One has
  $\psi(v_{k+1})= \psi(v_kx)$ with $x=a$ if $k$ is odd and $x=b$ if $k$ is even. In both cases one has
  $(v_{k+1})^- = v_k$ and $(v_{k+1})_+ = v_{k-1}$, so by Proposition \ref{prop:+-} and using the inductive hypothesis
$$ |a\psi(v_{k+1})b| = |a\psi(v_k)b| + |a\psi(v_{k-1})b| = F_{k+1}- F_{k-4} + F_k -F_{k-5} = F_{k+2}-F_{k-3}, $$
  which proves the assertion.
  \end{proof}

We say that a word $v\in {\cal A}^k$ is {\em almost alternating} if $v\in  [v_k]$.
  
  \begin{theorem}\label{prop:noalt} Let $k\geq 3$. If $v\in {\cal A}^k$ is not alternating, then 
  $$ |a\psi(v)b| \leq F_{k+1} - F_{k-4},$$
  where the upper bound is reached if and only if  $v$ is almost alternating.
  \end{theorem}
  \begin{proof} The proof is by induction on the integer $k$. If $k=3$, then if $ v\neq xyx$ with $\{x, y\}= \{a, b\}$ then
  $|a\psi(v)b| \leq 7 = F_4- F_{-1}$ and the upper bound is reached if and only if $v \in \{ab^2, b^2a, ba^2, a^2b\}$.
  If $k=4$ and $v$ is not alternating, then $|a\psi(v)b| \leq 12= F_5-F_0$ and the maximal value is reached if
  and only if  $v\in \{ab^2a, ba^2b\}$.
  
  Let us now consider the word $vxy$ with $x,y\in {\cal A}$ and $|vxy|= k >4$ and first prove that
  if $vxy$ is not alternating, then  $|a\psi(vxy)b| \leq F_{k+1} - F_{k-4}$. Let us first suppose that $y=x$. 
  
  If $v=x^{k-2}$,
  then $|a\psi(vxx)b| = |ax^kb|= k+2 < F_{k+1}-F_{k-4}$ and we are done.  If $v\neq x^{k-2}$, then $vxx$ is not constant, so that  by Proposition \ref{prop:+-}
  one has
  $$|a\psi(vxx)b| = |a\psi(vx)b| + |a\psi(v')b|,$$
  where $v'= (vxx)_+$ is a prefix of $v$ of length  $<  k-2$. By Lemma \ref{lemma:limit} one has $|a\psi(vx)b| \leq F_{k}$ and
  $|a\psi(v')b|\leq F_{k-2}$. Thus, since $k>4$, one has 
  \begin{equation} \label{eq:xx}
   |a\psi(vxx)b| \leq F_k + F_{k-2} = F_{k+1}-F_{k-3} <  F_{k+1}-F_{k-4}.
  \end{equation}
  
  Let us now suppose $x\neq y$. By Proposition \ref{prop:+-} one has
  \begin{equation}\label{eq:alter}
    |a\psi(vxy)b| = |a\psi(vx)b| + |a\psi(v)b|.
    \end{equation}
    Since $vxy$ is not alternating, so will be $vx$. By induction $|a\psi(vx)b| \leq F_k -F_{k-5}$. 
    
    If $v$ is not alternating,
    then, by induction,  $|a\psi(v)b| \leq F_{k-1}- F_{k-6}$. Hence, $|a\psi(vxy)b|\leq F_k -F_{k-5}+ F_{k-1}- F_{k-6} = F_{k+1}-F_{k-4}$
    and we are done.
    
  If  $v$ is alternating, then as $vx$ is not alternating, the only possibility is $v= (yx)^{\frac{k-2}{2}}$ if $k$ is even
    and $v= (xy)^{\lfloor \frac{k-2}{2}\rfloor}x$ if $k$ is odd. Hence, if $k$ is even, $vxy= y(xy)^{\frac{k-4}{2}}x^2y$ and
    if $k$ is odd, $vxy= (xy)^{\lfloor \frac{k-2}{2}\rfloor}x^2y$.  In both the cases  $vxy \in [v_k]$.
    By Lemma \ref{upp:1}, $|a\psi(vxy)b| = F_{k+1}-F_{k-4}$. Thus the first assertion is proved.
    
    In view of Lemma \ref{upp:1} it remains to prove that if $|a\psi(vxy)b| = F_{k+1}-F_{k-4}$, then $vxy \in [v_k]$. In this case, in view of (\ref{eq:xx}), necessarily
    $x\neq y$. If $v$ is alternating, as we have previously seen, $|a\psi(vxy)b| $ reaches its maximal value and $vxy \in [v_k]$. If $v$ is not alternating one has to require that both $|a\psi(vx)b|$ and $ |a\psi(v)b|$ reach their maximal values. By induction this occurs
    if and only if $vx\in [v_{k-1}]$ and $v\in [v_{k-2}]$. If $k$ is odd, the only possibility is $v= xy^2(xy)^{\frac{k-5}{2}}$, so that
    $vxy \in [v_k]$. If $k$ is even, then necessarily  $v= yx^2(yx)^{\lfloor \frac{k-5}{2} \rfloor}y$, so that also in this case $vxy \in [v_k]$.
  \end{proof}
  
  One easily verifies that for each $k \geq 3$, $\langle bv_kb \rangle = \frac{17\cdot 2^{k-1} + (-1)^{k+1}}{3}$, $\langle b{\bar v}_kb \rangle = \frac{19\cdot 2^{k-1} + (-1)^{k}}{3}$, $\langle bv^{\sim}_kb \rangle = \frac{7+ 2^{k}(9 + (-1)^{k+1})}{3}$, and $\langle b{\bar v}^{\sim}_kb \rangle = \frac{-7+ 2^{k}(9 + (-1)^{k})}{3}$. 
  
   By  the preceding theorem
  and Lemma \ref{lemma:limit}, one derives some identities on Stern's sequence. For instance,  for any $k\geq 3$ 
  
  $$ s\left ( \frac{17 \cdot 2^{k-1}+ (-1)^{k+1}}{3}\right ) =  s\left ( \frac{2^{k+3}+ (-1)^{k+2}}{3}\right )-  s\left ( \frac{2^{k-2}+ (-1)^{k-3}}{3}\right ).$$
  \begin{remark}
 {\em  Let us observe that setting  $R_k= \frac{17\cdot 2^{k-1} + (-1)^{k+1}}{3}$ for each $k\geq 1$, one has $R_1=6$ and $R_{k+1}= 2R_k +(-1)^k$. Similarly, if one defines inductively the three sequences $(S_k)_{k>0}$, $(T_k)_{k>0}$, and $(U_k)_{k>0}$ respectively as
  $S_1=6, T_1= 9, U_1=3$ and for $k\geq 1$
  $$ S_{k+1}= 2S_{k}+ (-1)^{k+1}, \ T_{k+1}= T_k + (3+(-1)^{k})2^k, \ \ U_{k+1}= U_k + (3+(-1)^{k+1})2^k ,$$
  one has that for $k\geq 3$, $\langle b{\bar v}_kb \rangle = S_k, \langle bv^{\sim}_kb \rangle = T_k,$ and $\langle b{\bar v}^{\sim}_kb \rangle=U_k$.}
  \end{remark}

 For any $k\geq 0$ and $n\geq 0$ let  $C_k(n)$ denote the number of Christoffel words of length $n$ and order $k$.  By  \cite[Lemma 5]{DM}, and since the
palindromization map is injective, one has
\begin{equation}\label{eq:cicappa}
\sum_{k\geq 0} C_k(n) = \phi(n)  \ \mbox{and} \ \  \sum_{n \geq 0} C_k(n) = 2^k,
\end{equation}
where $\phi$ is the Euler  totient  function. By  Lemma \ref{lemma:limit} one has that  $C_k(n)=0$ for  $n<k+2$ and for $n > F_{k+1}$.
Moreover, from Proposition \ref{prop:noconst} and  Theorem~\ref{prop:noalt}  one has that
$C_k(n)= 0$ for $k+2<n < 2k+1$ and for  $F_{k+1}- F_{k-4} <n < F_{k+1}$.

For each $k\geq 0$ we introduce the set of the {\em missing lengths} of order $k$
$$ ML_k = \{ n \mid  k+2 \leq n \leq F_{k+1} \  \mbox{and} \  C_k(n) = 0\}. $$
The first values of $\card(ML_k)$, $0\leq k \leq 20$, are reported below
$$0, 0, 1, 2, 5, 11, 18, 29, 51, 74, 119, 195, 323, 498, 828, 1361, 2289,3801, 6305, 10560. $$

\noindent By  Proposition \ref{prop:noconst} and Theorem \ref{prop:noalt}  one has that for $k\geq 3$
$$\card (ML_k) \geq F_{k-4} + k-3.$$
Since for large $k$, one has $F_{k-4} \approx  \frac{g^{k-2}}{\sqrt 5}$, where $g$ is the
 golden number  $g = \frac{1+ \sqrt 5}{2}= 1.618\cdots$, it follows
that the lower bound to $\card (ML_k)$ is exponentially increasing with $k$.
Moreover, one easily derives that  $$\liminf_{k\rightarrow \infty} \frac{\card (ML_k)}{F_{k+1}} \geq  \frac{1}{g^5}.$$

For each $k\geq 0$  let  $L_k$ denote the set of all  lengths  of  Christoffel words of order $k$. 

 The following lemma shows  that for  $k \geq 2$ there exist  Christoffel words of order $k$ having 
lengths which are consecutive integers. Some examples are given in the following lemma.
\begin{lemma} For $k\geq 2$ one has $ 3k-2, 3k-1 \in L_k$ and for $k\geq 3$, one has $5k-8, 5k-7 \in L_k$.
\end{lemma}
\begin{proof} One easily verifies that for $k\geq 2$, one has
$|a\psi(a^2b^{k-2})b| = 3k-2$ and $|a\psi(aba^{k-2})b|= 3k-1$.  For $k\geq 3$, one has $|a\psi(ab^2a^{k-3})b|= 5k-8$ and $|a\psi(abab^{k-3})b|= 5k-7$. 
\end{proof}

For each $k\geq 1$, we set
$$ M_k = \max \{C_k(n) \mid n\geq 0\}.$$
The first values of $M_k$, $0< k \leq 22$ and the values $n_k$ for which $C_k(n_k)= M_k$ 
 are reported in the following table:

\vspace{10 mm}

\begin{center}

 \begin{tabular}{l|l|ll||l|ll|ll}
 $ k $     &  $M_k $ & $n_k$ & &  $k$ & $M_k$ & &$n_k$ & \\ \hline
1 & 2& 3 &  &12 &36 & & 199, 283\\
2 &  2 & 4,5 & &13 & 48 & & 449 \\
3 & 4& 7 & & 14 & 64 & & 433\\
4 & 4& 9, 11& & 15 & 72 & & 839\\
5 & 4& 11,13,14,17,18,19& &16 & 102 & & 1433\\
6 & 8 & 23 & &17& 124 & & 1997\\
7 &12& 41& & 18 & 160 & & 1987\\
8 & 12 & 43 & & 19 & 212 & & 3361\\
 9 &16& 71,73,83 & & 20 &256 & & 5557\\
 10 & 24& 113 & & 21& 332& & 8689\\
 11 &28 &  227& & 22 & 444 && 8507
\end{tabular}

\end{center}

\vspace{10 mm}

\begin{proposition}  $\lim_{k\rightarrow \infty} M_k = \infty. $
\end{proposition}
\begin{proof} By Lemma \ref{lemma:limit} and (\ref{eq:cicappa}) one has
$$ 2^k =  \sum_{n \geq 0} C_k(n) = \sum_{n \geq k+2}^{F_{k+1}} C_k(n) \leq M_k F_{k+1}.$$
By Binet's formula of Fibonacci numbers one has that  $F_{k+1} < g^{k+3}$. Hence, $M_k \geq  \frac{1}{g^3} ( \frac{2}{g})^k$.
From this the result follows.
\end{proof}
From the proof of previous proposition one has that  $M_k$ has a lower bound which is exponentially increasing with $k$, whose values are much less than those of $M_k$ given in the table above. It would be interesting to find tight lower and upper bounds for $M_k$ and possibly a formula to compute its values and also the values of the lengths of Christoffel words for which $C_k(n)$ is equal to $M_k$. Moreover, from the table one has that for all $0<k<22$, $M_k$ is non-decreasing with $k$ and  $M_{k+1} \leq M_k+ M_{k-1}$. We conjecture that this is true
for all  $k$.


\small

  \begin{thebibliography}{99}

 \bibitem{JB} J. Berstel, Sturmian and Episturmian Words (A Survey of Some Recent Results), in S. Bozapalidis and G. Rahonis (Eds.): CAI 2007, LNCS 4728, pp.23--47, Springer-Verlag Berlin 2007
\bibitem{BDL} J. Berstel, A. de Luca, Sturmian words, Lyndon words and trees, \emph{Theor. Comput. Sci.} \textbf{178} (1997) 171--203
\bibitem{BLRS} J. Berstel, A. Lauve, C. Reutenauer, F. V. Saliola, {\em Combinatorics on Words, Christoffel Words and Repetitions in Words}, CRM Monograph series, vol. 27, American Mathematical Society, (Providence, RI, 2009)
\bibitem{BDR} V. Berth\'e, A. de Luca, C. Reutenauer, On an involution of Christoffel words and Sturmian morphisms, {\em Eur.  J.  Combin.}  {\bf 29} (2008) 535--553
\bibitem{BL} J.-P. Borel, F. Laubie, Quelques mots sur la droite projective r\'eelle, {\em J. Th\'eor.  Nombres  Bordeaux} {\bf 5} (1993) 23--52
\bibitem{CW} N. Calkin, H. Wilf, Recounting the rationals, {\em Amer. Math. Montly} {\bf 107} (2000) 360--363
\bibitem{CW1} N. Calkin, H. Wilf, Binary partitions of integers and Stern-Brocot-like trees, 2009 (unpublished)
\bibitem{CdL} A. Carpi, A. de Luca, Harmonic and gold Sturmian words, {\em Eur.  J.  Combin.} {\bf 25} (2004) 685--705
\bibitem{CFF} E. B. Christoffel, Observatio arithmetica, {\em Ann.  Mat. Pur. Appl.} {\bf 6} (1875) 148--152
\bibitem{CS} M. Coons, J. Shallit, A pattern sequence approach to Stern's sequence, {\em Discrete Math.} {\bf 311}(2011) 2630--2633
\bibitem{deluca} A. de Luca, Sturmian words: Structure, Combinatorics, and their Arithmetics, \emph{Theor. Comput. Sci.} \textbf{183} (1997) 45--82
\bibitem{SC} A. de Luca,  A standard correspondence on epicentral words, {\em Eur.  J.  Combin.} {\bf 33} (2012) 1514--1536
\bibitem{FiDE} A. de Luca, Some extremal properties of the Fibonacci words, {\em Internat.  J.  Algebra Comput.} {\bf 23} (2013) 705--728
\bibitem{adlADL} A. de Luca, A. De Luca, Pseudopalindrome closure operators in free monoids, 
\emph{Theor. Comput. Sci.} \textbf{362} (2006) 282--300
\bibitem{adlADL1} A. de Luca, A. De Luca, A generalized palindromization map in free monoids, \emph{Theor. Comput. Sci.} {\bf 454} (2012)109--128
\bibitem{DM} A. de Luca, F. Mignosi, Some combinatorial properties of Sturmian words, \emph{Theor. Comput. Sci.} \textbf{136} (1994) 361--385
\bibitem{adlZ1} A. de Luca, L. Q. Zamboni, On graphs of central episturmian words, \emph{Theor. Comput. Sci.} {\bf 4} (2010) 70--90
\bibitem{adlZ2} A. de Luca, L. Q. Zamboni, Involutions of epicentral words, \emph{Eur.  J.  Combin.} {\bf 31} (2010) 867--886
\bibitem{ADLReut} A. De Luca, C. Reutenauer. Christoffel words and the Calkin-Wilf tree, {\em  Electron. J. Combin. } {\bf 18} (2) (2011), P22
 \bibitem{DJP} X. Droubay, J. Justin,  G. Pirillo, Episturmian words and some constructions of de Luca and Rauzy,  \emph{Theor. Comput. Sci.} \textbf{255} (2001) 539--553
\bibitem{GKP} R. L. Graham, D. E. Knuth, O. Patashnik, {\em Concrete Mathematics}, 2-nd edition,
Addison-Wesley (Reading Mass., 1994)
\bibitem{CMRJ} J. Justin, On a paper by Castelli, Mignosi, Restivo,  \emph{RAIRO-Theor. Inform. Appl.} \textbf{34} (2000) 373--377
\bibitem{J} J. Justin, Episturmian morphisms and a Galois theorem on continued fractions,  \emph{RAIRO-Theor. Inform. Appl.} \textbf{39} (2005) 207--215
\bibitem{KREU} C. Kassel, C. Reutenauer, A palindromization map for the free group, \emph{Theor. Comput. Sci.}  {\bf 409} (2008) 461--470
\bibitem{Leh} D. H. Lehmer,  On Stern's diatomic sequence, {\em Amer. Math. Montly} {\bf 36} (1929) 59--67
\bibitem{LO} M. Lothaire, \emph{Combinatorics on Words}, Addison-Wesley (Reading, MA, 1983)
\bibitem{LO2} M. Lothaire, \emph{Algebraic Combinatorics on Words}, Encyclopedia of Mathematics and its Applications, vol.~90, Cambridge University Press (Cambridge, 2002)
\bibitem{EL} E. Lucas, {\em Th\'eorie des nombres}, Gauthier-Villars (Paris, 1891)
\bibitem{MH} M. Morse, G. A. Hedlund, Symbolic dynamics II: Sturmian sequences, {\em Amer.  J.  Math.}  {\bf 61} (1940) 1--42
\bibitem{North} S. Northshield, Stern's diatomic sequence 0,1,1,2,1,3,2,3,1,4,..... , {\em Amer. Math. Montly} {\bf 117} (2010) 581--598
\bibitem{Ra} G. N. Raney, On continued fractions and finite automata, {\em Math. Annal.}  {\bf 206} (1973) 265--283
\bibitem{Sloane} N. J. A. Sloane, (2008), The on-line Encyclopedia of integer sequences, www.research.att.com/njas/sequences/ 
\bibitem{Stern} M. A. Stern, \" {U}ber eine zahlentheoretische Funktion, {\em J. Reine Agnew. Math.} {\bf 55} (1858) 193--220
\bibitem{IU} I. Uriba, Some properties of a function studied by De Rham, Carlitz and Dijkstra and its relation to the (Eisenstein-)Stern's diatomic sequence, {\em Math. Commun.} {\bf 6} (2001) 181--198
\end{thebibliography}
\end{document}